\documentclass[11pt,leqno]{amsart}

\usepackage{amsmath}
\usepackage{amssymb}
\usepackage{a4wide}
\usepackage{bbm}
\usepackage{graphics}
\usepackage{epsfig}
\usepackage{amsmath, amssymb, latexsym, amscd, amsthm,amsfonts,amstext}
\usepackage[mathscr]{eucal}
\usepackage{appendix}
\usepackage[active]{srcltx}%goi lenh de sua tu dvi sang tex
\include{srctex}%goi lenh de sua tu dvi sang tex

\usepackage{color}

 \usepackage{mleftright}

\newcounter{hypo}

\makeatletter
 \@addtoreset{equation}{section}
 \makeatother
 
 \usepackage{hyperref} 

\newtheorem{theorem}{Theorem}[section]
\newtheorem{lemma}[theorem]{Lemma}
\newtheorem{proposition}[theorem]{Proposition}
\newtheorem{definition}[theorem]{Definition}

\theoremstyle{definition}

\numberwithin{equation}{section}
 
  {\setlength{\baselineskip}{1.5\baselineskip}

\title[Trace formula and Spectral shift  Function]
{ S\MakeLowercase{emiclassical} T\MakeLowercase{race} F\MakeLowercase{ormula} 
  \MakeLowercase{and} S\MakeLowercase{pectral} S\MakeLowercase{hift} F\MakeLowercase{unction}\\
\MakeLowercase{for} S\MakeLowercase{ystems}  \MakeLowercase{via a }S\MakeLowercase{tationary}  A\MakeLowercase{pproach}}
 
 \author[M. Assal, \; M. Dimassi \;and \; S. Fujii\'e]{Marouane Assal, Mouez Dimassi and Setsuro Fujii\'e}
\address{Marouane Assal, Mouez Dimassi, IMB (UMR-CNRS 5251), UNIVERSIT\'E DE BORDEAUX, 351 COURS DE LA LIB\'ERATION,
33405 TALENCE CEDEX, FRANCE}
\email{marouane.assal@math.u-bordeaux.fr}
\email{mdimassi@u-bordeaux.fr}
\address{Setsuro Fujii\'e, RITSUMEIKAN UNIVERSITY, 1-1-1 NOJI-HIGASHI, 525-8577 KUSATSU, JAPAN}
\email{fujiie@fc.ritsumei.ac.jp}

\subjclass[2010]{81Q10 (47A55, 81Q20, 47N50)}
\keywords{Spectral shift function, matrix Schr\"odinger operators, asymptotic expansions}

\begin{document}

\maketitle
 
%\tableofcontents

%\vfill\break

\begin{abstract}
We establish a semiclassical trace formula in a general framework of microhyperbolic hermitian systems of $h$-pseudodifferential operators, and  apply  it to the study of the spectral shift function associated to a pair of selfadjoint Schr\"odinger operators with matrix-valued potentials.  We give  Weyl type semiclassical asymptotics with sharp remainder estimate for the spectral shift function, and, under the existence of a scalar escape function, a full asymptotic expansion in the strong sense for its derivative.
A time-independent approach enables us to treat certain potentials with energy-level crossings. 
\end{abstract}

\section{Introduction}

In this paper, we study the spectral shift function (SSF for short) for Schr\"odinger operators with matrix-valued potentials. Such operators appear in molecular physics in the Born-Oppenheimer approximation. The justification of this approximation and a classification of matrix Schr\"odinger operators can be found in \cite{Com, Hag, Jecko, Klein}.

More precisely, we are concerned with the SSF for the pair of operators $(P_1, P_0)$ with 
\begin{equation}\label{Studied operators}
P_0:= -h^2 \Delta \otimes I_N + V_{\infty}, \quad P_1:= -h^2 \Delta \otimes I_N + V(x),
\end{equation}
where $h \in (0,1]$ is a small positive parameter, $I_N$ is the identity $N\times N$ matrix, $V_{\infty}$ is an $N\times N$ constant hermitian matrix and $V(x)$ is a smooth  hermitian matrix-valued potential which tends rapidly enough to $V_{\infty}$ at infinity. 
The SSF {associated to $(P_1,P_0)$} denoted $s_h$ is defined as distribution (modulo a constant) by the Lifshits-Krein formula
\begin{equation}\label{LK}
\langle s_{h}',f\rangle = - {\rm tr}\big( f(P_1)-f(P_0) \big), \quad \forall f\in C_0^{\infty}(\mathbb R;\mathbb R).
\end{equation} 
The SSF is related with the eigenvalue counting function below the level $\inf \sigma (P_0)$ and with the scattering determinant above this level (Birman-Krein formula, {see \cite{Yafaev}}). {Here $\sigma(P_0)$ stands for the spectrum of $P_0$.}

The concept of the SSF was introduced in the middle of the previous century by I. M. Lifshits in his investigations in the solid state theory (see \cite{Lif, Lif1}) and then developed by M. Krein (see \cite{Krein3, Krein1, Krein2}) into a mathematical theory.
%a powerful tool which can be used in many applications including spectral asymptotics, scattering theory, traces formulas for self-adjoint operators, semiclassical approximation, resonance theory ... 
The work of Krein on the SSF has been described in details in the survey \cite{Birman2}. One can also find detailed account concerning mathematical and historical aspects of the SSF in \cite{Birman1}.

In the scalar case $N=1$, a lot of works have been devoted to the study of the  SSF
{in different asymptotic regimes} (see \cite{Robert1} and the references therein).
In particular, a Weyl-type asymptotics of the SSF with a sharp remainder estimate and a complete asymptotic expansion of the derivative of the SSF 
were studied in high energy regime (\cite{Robert2}) and in the semiclassical regime (\cite{Rob}, \cite{Robert}).

The proofs of these works reduce to the study of 
\begin{equation}
\label{tr}
{\rm tr}\left (f(P_1){\mathcal F}_h^{-1}\theta (\tau-P_1)-f(P_0){\mathcal F}_h^{-1}\theta (\tau-P_0)\right ),
\end{equation}
where $\theta$ is a smooth function of the time $t$  with compact support and ${\mathcal F}_h^{-1}$ is the  semiclassical Fourier inverse transform
defined by \eqref{semifourier}. The method in \cite{Rob} consists in writing \eqref{tr} as the semiclassical Fourier inverse transform of
$\theta(t)\,{\rm tr}\left ( f(P_1)e^{-itP_1/h}-f(P_0)e^{-itP_0/h}\right )$, and constructing (modulo ${\mathcal O}(h^\infty)$)  the Schwartz' kernel   of the evolution  operator % asymptotic kernel of the time evolution operator
$f(P_1)e^{-itP_1/h}$. %the  solution of the hyperbolic equation
%$$(hi\partial_t-P_1)I(t)=0,\,\,\, I(0)=f(P_1),  t\in ]-T,T[.$$
This construction by means of Fourier integral  operators  is now standard and well known  for scalar-valued  operators $P_1$ (see \cite{Hormander2, iv} for problems concerning  the  asymptotic distribution of eigenvalues,  and \cite{Robert1, Robert2, Rob} for the SSF). For  matrix-valued  operators this explicit  construction  is  very complicated (or  impossible).  To avoid this problem, and to study the counting function of eigenvalues of $P_1$, 
 V. Ivrii \cite{iv}
observed that a rough construction by  using  the successive approximation method of $f(P_1)e^{-itP_1/h}$ for  $\vert t\vert <h^{1-\delta}$ (with  $0<\delta\leq 1$)  suffices 
to get a full asymptotic expansion in powers of $h$ of  ${\rm tr}\,(f(P_1){\mathcal F}_h^{-1}\theta (\tau-P_1))$. This beautiful observation is used by    the second author and J. Sj\"ostrand  \cite{dimassi}   to develop  a time-independent approach  to get asymptotics of ${\rm tr}\,(f(P_1){\mathcal F}_h^{-1}\theta (\tau-P_1))$ for matrix-valued operator $P_1$. The novelty in this approach consists in expressing \eqref{tr} in terms
of the resolvent instead of evolution operator, and studying the (almost) analyticity of its trace near the real axis. This method is used in \cite{Dimassi2} to study the SSF for  scalar non semi-bounded operators such as Stark Hamiltonian. The aim of this paper is to develop and apply  this stationary approach to the study of the SSF for  matrix-valued  operators.

%results on trace asymptotics with small remainder estimates. 

% the time-dependent Schr\"odinger equation.

%The  study SSF of the matrix case $N>1$ can be reduced to the scalar case by diagonalization when the potential has uniformly distinct eigenvalues ({see \cite{BR} for the case of the Dirac operator).} However, in the general case where the eigenvalues cross, such reduction is no longer possible and the time evolution method becomes quite difficult to apply.
%We propose here to employ a  time-independent approach to microhyperbolic systems, which consists in expressing \eqref{tr} in terms
%of the resolvent instead of evolution operator, and studying the (almost) analyticity of its trace near the real axis.
%The idea of this approach will be explained in more detail in the outline of the proof.
%This method was first introduced by the second author and Sj\"ostrand in \cite{dimassi} for the study of the eigenvalue counting function and developed by the second and third authors in \cite{Dimassi2} for the study of the SSF corresponding to the stark Hamiltonian.

In the first part of this work, we consider a general system of $h$-pseudodifferential operator $H^w=H^w(x,hD_x)$.
For a fixed energy $\tau_0$  such that 
%and a cutoff function $\chi$ in the phase space, on the support of which
$\tau_0-H(x,\xi)$ is uniformly microhyperbolic in some direction $T$ (see Definition \ref{de1}),
we show that the trace of the operator $\chi^wf(H^w){\mathcal F}_h^{-1}\theta (\tau-H^w)$ is negligible 
($={\mathcal O}(h^\infty)$)
provided that $\theta$ is supported in $h^{1-\delta}\le \vert t\vert\leq \kappa$ (for arbitrary positive $h$-independent  $\kappa$), see  Theorem \ref{th1}. Here  $\chi\in C^\infty_0(\mathbb R^{2n};\mathbb R)$  and $f$ is supported in a small neighborhood of $\tau_0$. Moreover,  under the existence of an escape function associate to $H(x,\xi)$ at $\tau_0$ (see \eqref{deff}),  we can take $\kappa=h^{-\nu}$ for arbitrary  $ \nu>0$ (see  Remark 3.1 and section 4.4). On the other hand, we give a complete 
asymptotic expansion in  powers of $h$  of ${\rm tr}(\chi^wf(H^w){\mathcal F}_h^{-1}\theta (\tau-H^w))$ provided  that $\theta$ is supported in a
small $h$-independent neighborhood of $0$  and $\tau_0-H(x,\xi)$ is microhyperbolic at every point $(x,\xi)\in {\rm supp} \chi$, see Theorem \ref{WASG}.
This is a consequence from the fact that the above trace depends, modulo ${\mathcal O}(h^\infty)$, 
only on the symbol $\tau_0-H(x,\xi)$ on the support of $\chi$ as long as the support of $\theta$ is small enough near 0
 (Theorem \ref{th2}), 
and the fact that a symbol $\tau_0-H$ microhyperbolic near a point can be extended
to a uniformly microhyperbolic symbol in the whole phase space (Theorem \ref{main result appendix}). %for which the contribution to the above tracefrom time larger than $h^{1-\delta}$ with $\delta>0$ is ${\mathcal O}(h^\infty)$ (Theorem \ref{th1}).

To our best knowledge, there are only few works treating the semiclassical asymptotics  of the  SSF  for matrix valued  operators (see \cite{BR, Khoc} and the references  therein).  The  asymptotics of the SSF for the semi-classical  Dirac  operator has been  studied in \cite{BR}. In this case, the classical corresponding Hamiltonian 
 has uniformly distinct eigenvalues,  and then  the study of the SSF can be reduced to the scalar case by diagonalization.  The relation between the spectral shift function  and the resonances for  Dirac operator with analytic potential has been examined in \cite{Khoc}.
In the second part of this paper, we consider the SSF {associated to the pair of Schr\"odinger operators with matrix-valued potentials defined in \eqref{Studied operators}},  without any condition on the multiplicities of its eigenvalues. First, using Theorem \ref{WASG}, we show that \eqref{tr} has a full asymptotic expansion in $h$
when the support of $\theta$ is close enough to the origin (Theorem \ref{ASW2}). This result with a Tauberian argument give
the Weyl-type asymptotic formula for the SSF with a sharp remainder estimate (Theorem \ref{WRA}).
Finally we give a pointwise full asymptotic expansion of the derivative of the SSF {near} energies $\tau$ where there exists 
a scalar escape function {associated to the classical Hamiltonian $\xi^2 I_N + V(x)$} (Theorem \ref{main semiclassical}). This last theorem is a generalization to the matrix case
of the result of \cite{Rob} at non-trapping energies. 

The paper is organised as follows. In section \ref{Statement of the results}, we state our main results and we give an outline of the proofs. The proofs of these results will be given in Sections \ref{Proofs of the results STF} and \ref{Proofs of the results SSF} respectively. Finally, the appendix \ref{fonctions microhyperboliques} contains some technical lemmas related to the notion of microhyperbolicity used in our proofs.

\textbf{Notations :} For $\xi = (\xi_1,...,\xi_n)\in \mathbb R^n$, we use the usual notation $\langle \xi \rangle := (1+\xi^2)^{1/2}$, where $\xi^2:= \xi_1^2+...+\xi_n^2=\vert \xi\vert^2$. For $z\in \mathbb C$, we recall that $\bar{\partial}_z := \frac{1}{2}(\partial_{\Re z} + i \partial_{\Im z})$. The bracket $[a_j]_0^1$ stands for the difference $a_1-a_0$. The scalar products in $\mathbb R^n$ and $\mathbb C^N$ will be denoted $\langle \; , \; \rangle$ and $(\;,\;)$ respectively. We introduce the following standard asymptotic notations that we shall use through the paper. Given a function $f_{h}$ depending on a small parameter $h\in(0,1]$, the relation $f_{h}=\mathcal{O}(h^{\infty})$  (or $f_h\equiv 0$) means that $f_{h}=\mathcal{O}(h^k)$, for all $k\in \mathbb N$ and $h$ small enough. We write $f_{h}\sim \sum_{j\geq 0}\gamma_jh^j$ provided that for each $k\in \mathbb N$, $f_{h}-\sum_{j=0}^k \gamma_j h^j=\mathcal{O}(h^{k+1})$.

\section{Statement of the results}\label{Statement of the results}

 Let $\mathcal{H}_N$ be the space of hermitian $N\times N$ matrices endowed with the norm $\Vert \cdot \Vert_{N\times N}$, where for $A\in \mathcal{H}_N$, $\Vert A \Vert_{N\times N}:= \sup_{\{w\in \mathbb R^N;\; |w|<1 \}} | A w |$. 

Throughout this work we will use the notations of \cite{dimassi} for symbols and $h$-pseudodifferential
operators (see also \cite{iv}). In particular, $S^0(\mathbb R^{2n};\mathcal{H}_N)$ is the class of symbols 
$$
S^0(\mathbb R^{2n};\mathcal{H}_N):=\{H\in C^\infty({\mathbb R}^{2n}; \mathcal{H}_N); \, 
\Vert \partial_x^{\alpha} \partial_{\xi}^{\beta} H(x,\xi)\Vert_{N\times N} ={\mathcal O}_{\alpha,\beta} (1),\,\, \forall \alpha,\beta\}.
$$

We use the standard Weyl quantization of symbols. More precisely, if
$H \in S^0(\mathbb R^{2n};\mathcal{H}_N)$ then $H^w(x,hD_x)$ is the operator
 defined by  
$$
H^w(x,h D_x) u (x) = \frac{1}{(2\pi h)^n} \int\int_{\mathbb R^{2n}} e^{i(x-y)\cdot\xi/h} H\left( \frac{x+y}{2},\xi \right) u(y) dyd\xi, \quad u\in C^\infty_0(\mathbb R^n;\mathbb C^N).
$$
We will occasionally use the shorthand notations ${\rm Op}^w_h(H)=H^w=H^w(x,hD_x)$ when there is no ambiguity.

We recall the following notion of microhyperbolicity which will play an important role in this paper. 

\begin{definition}[Microhyperbolicity] \label{de1} \normalfont Let $H\in C^{\infty}(\mathbb R^{2n};\mathcal{H}_N)$. 
 We say that $H(x,\xi)$ is micro-hyperbolic at $(x_0,\xi_0)$  in the
direction $T\in {\mathbb R}^{2n}$, if there are constants $C_0,C_1,C_2>0$  such that
\begin{equation}\label{MHC}
\big(\langle T,\nabla_{x,\xi} H(x,\xi) \rangle w , w \big) \geq C_0  |w|^2 - C_1 |H(x,\xi)w|^2,
\end{equation} 
for all $(x,\xi)\in {\mathbb R}^{2n}$ with $\vert(x,\xi)-(x_0,\xi_0)\vert\leq\frac{1}{C_2}$ and all $w \in \mathbb C^N$. {Here $\nabla_{x,\xi}H(x,\xi)=\left(\partial_xH(x,\xi),\partial_{\xi}H(x,\xi)\right)$}. If for some constants $C_0,C_1>0$  the above
estimate holds for all $(x,\xi)\in {\mathbb R}^{2n}$,  we say that  $H(x,\xi)$ is uniformly
microhyperbolic on ${\mathbb R}^{2n}$  in the direction $T$.
In the case where $H(x,\xi)$ depends also on an additional  parameter,  we say that $H$ is uniformly microhyperbolic   in the direction $T$ if \eqref{MHC} is satisfied with $C_0,C_1>0$ independent of  this parameter.
\end{definition}

\subsection{{Trace formula for systems of $h$-pseudodifferential operators}} 

Let 
$$
\theta\in C_0^{\infty}(]-1,1[;\mathbb R),\quad \theta_{\varepsilon}(t):= \theta( {t}/{\varepsilon} ),
$$
where $\varepsilon>0$ is a positive constant possibly depending on $h$ and
\begin{equation}
\label{semifourier}
\mathcal{F}_{h}^{-1}\theta_{\varepsilon}(\tau)= \frac{1}{2\pi h}\int_{\mathbb R} e^{it\tau/h}\theta_{\varepsilon}(t)dt,
\end{equation}
the semiclassical Fourier inverse operator.

%Fix two real numbers $a<b$ and let $f(\tau)\in C_0^{\infty}(]a,b[;\mathbb R)$, 

Let $A, H\in S^0(\mathbb R^{2n};\mathcal{H}_N),$  and $\chi\in C_0^{\infty}(\mathbb R^{2n};\mathbb R)$.   We assume that
$A^w(i+H^w)^{-k}$ is of trace class for some $k\in {\mathbb N}$. Writing $A^w f(H^w) = A^w(i+H^w)^{-k} (i+ H^w)^k f(H^w)$ and using the fact that $(i+ H^w)^k f(H^w)$ is bounded by the spectral theorem we deduce that $A^w f(H^w)$ is of trace class for all $f\in C_0^{\infty}(\mathbb R;\mathbb R)$. We recall that $\chi^w$ is of trace class (with norm trace $\mathcal{O}(h^{-n})$, see \cite[Theorem 9.4]{dimassi}).

Fix $\tau_0\in \mathbb R$. We denote by
${O}_{\tau_0}$ the set of open intervals centered at $\tau_0$, i.e.,
$$O_{\tau_0}=\{]\tau_0-\eta,\tau_0+\eta[;\;\; \eta>0\}.$$
%We define 
%\begin{equation}
%G_H(\lambda,\varepsilon;h) :=  {\rm tr}\; \bigg[  \chi^w(x,h D_x) f\big(H^w(x,h D_x)\big) \mathcal{F}_{h}^{-1}\theta_{\varepsilon}\big(\lambda-H^w(x,h D_x)\big) \bigg].
%\end{equation}
%In the following $\varepsilon$ will be an $h$-dependent positive function. When $\varepsilon=1$, we shall simply write $G_H(\lambda;h)$ instead of $G_H(\lambda,1;h)$.

\begin{theorem}\label{th1} 
Suppose that there exists $T\in \mathbb R^{2n}$ such that $\tau_0-H(x,\xi)$
 is uniformly   
micro-hyperbolic with respect to $(x,\xi)\in {\mathbb R}^{2n}$ in the direction $T$. 
If $0\notin {\rm supp}\,\theta$, then there exists $I\in O_{\tau_0}$ such that for all 
$f\in C^\infty_0(I;\mathbb R)$ and $\varepsilon\in [h^{1-\delta},\kappa[$ with  $\kappa>0$, $0<\delta\le 1$ independent of $h$,  we have, uniformly for $\tau\in \mathbb R$,
\begin{equation}\label{main estimate 1}
{\rm tr}\; \left(  A^wf\big(H^w) \mathcal{F}_{h}^{-1}\theta_{\varepsilon}\big(\tau-H^w) \right)= \mathcal{O}(h^{\infty}).
\end{equation}
\end{theorem}

\begin{theorem}\label{th2}
Let $H_0,H_1\in S^0(\mathbb R^{2n},\mathcal{H}_N)$ be such that $H_0= H_1$ {in a neighborhood of} ${\rm supp}\,\chi$. 
Then there exists $\varepsilon>0$ small and independent of $h$ such that we have, uniformly for $\tau \in \mathbb R$,
\begin{equation}\label{man2}
{\rm tr}\; \left(  \chi^w \left[  f (H_j^w) \mathcal{F}_{h}^{-1}\theta_\varepsilon (\tau-H_j^w) \right]_0^1\right) = \mathcal{O}(h^{\infty}).
\end{equation}
 \end{theorem}

The following result is a simple consequence of the above theorems.
\begin{theorem}\label{WASG}
Suppose that  $\tau_0-H(x,\xi)$ is microhyperbolic at every point  $(x,\xi)$ in ${\rm supp}\,\chi$. 
If $\theta$ equals 1 near $t=0$, then there exist $I\in O_{\tau_0}$ and  $\varepsilon>0$ small and independent of $h$
such that for $f\in C^\infty_0(I;\mathbb R)$, the following full asymptotic expansion in powers of $h$ holds
uniformly for $\tau\in \mathbb R$: 
\begin{equation}\label{man3}
{\rm tr}\; \left(  \chi^w  f (H^w) \mathcal{F}_{h}^{-1}\theta_\varepsilon (\tau-H^w)  \right) \sim (2\pi h)^{-n} f(\tau) \sum_{j\geq 0}  \gamma_{j}(\tau) h^{j} \quad \text{as}\;\;h \searrow 0.
\end{equation} 
\end{theorem}

{\bf Remark 2.2.} The coefficients $\tau \mapsto \gamma_j(\tau)$ are smooth, {independent of $f$ and $\theta$} and can be computed explicitly (see  formula \eqref{D7}). 
\subsection{Application to Schr\"odinger operators with matrix-valued potentials}\label{SSF}
%In this paragraph, we state our main results concerning the spectral shift function associated to a pair of self-adjoint systems of $\h$-pseudodifferential operators. Relying on the results of the previous paragraph, our application for the study of the SSF apply to general models of systems of $\h$-pseudodifferential operators. In order to simplify the presentation and to avoid technicalities in the statement of the needed assumptions, we restrict our selves to Schr\"odinger operators with matrix-valued potentials.

In this section we apply the above trace formula  to study the spectral properties of multi-channel semiclassical  Schr\"odinger operators of the form 
\begin{equation}\label{Operators}
P_1(h) :=-h^2\Delta \otimes I_N  + V(x),\,\,\,\, P_0(h) :=-h^2\Delta \otimes I_N +V_\infty,\,\,\,\text {in } 
L^2(\mathbb R^n;\mathbb C^N),
\end{equation}
where $I_N$ is the $N\times N$ identity matrix and $V(x)$ is a smooth hermitian matrix-valued potential, i.e.,
$$V(x)=\big (V_{ij}(x)\big)_{1\leq i,j\leq N},\,\, V_{ij}(x)=\overline{V_{ji}(x)}.$$
%multiplication operator by $V\in C^{\infty}(\mathbb R^n;H_N)$ satisfying the following %assumption 
We assume that the matrix $V$ has a limit $V_\infty$ at infinity and 
%$\alpha,\beta \in {\mathbb N}^n$, $\exists C_{\alpha,\beta}>0$ such that
\begin{equation}\label{DA1}
\exists \mu >n\,\,\,\, {\rm s.t.}\,\,\,\, \Vert \partial_x^{\alpha} (V(x)-V_\infty)\Vert_{N\times N} ={\mathcal O}_\alpha (\langle x \rangle^{-\mu-|\alpha|}),\,\,\,\,\,\, \forall \alpha\in \mathbb N^n,\, \forall x\in {\mathbb R}^n.
\end{equation}
After a linear transformation, we may assume that
%After a linear change of coordinates we may assume that
$$V_\infty=\begin{pmatrix}
e_{1,\infty} & 0&\cdots&0 \\
0& e_{2,\infty}&\cdots&0\\
\vdots&0& \ddots&\vdots\\
0&\cdots&0&e_{N,\infty}
\end{pmatrix},\,\,\,\text {with } e_{1,\infty}\leq e_{2,\infty}\leq \cdots \leq e_{N,\infty}.$$

The operator $P_0(h)$ with domain $H^2({\mathbb R}^n;{\mathbb C}^N)$ is self-adjoint. Its spectrum is $[e_{1,\infty},+\infty[$.
Since $V-V_\infty$ is $\Delta$-compact, the operator $P_1(h)$ admits a unique self-adjoint
 {realization} in $L^2({\mathbb R}^n;{\mathbb C}^N)$ with domain $H^2({\mathbb R}^n;{\mathbb C}^N)$. Moreover the essential  spectra of $P_1(h)$ and $P_0(h)$ are the same. The operator $P_1(h)$ may have discrete eigenvalues in $(-\infty,e_{1,\infty})$ and embedded ones in the interval $[e_{1,\infty},e_{N,\infty}]$ contained in the continuous spectrum.

%Set $P_0(h):=-h^2 \Delta + V_{\infty}$. Under this assumption, the operators $P_0(h),P_1(h)$ are self-%adjoint on $\big(H^2(\mathbb R^n)\big)^N := H^2(\mathbb R^n) \times ... \times H^2(\mathbb R^n).$

The spectral shift function $s_h(\tau)$ {associated to $(P_1(h),P_0(h))$} is defined as a real-valued
function  on ${\mathbb R}$ satisfying the Lifshits-Krein formula 
%Assumption \textbf{(A1)} also ensures that the operator $(P_0(h)+i)^{-k}-(P_1(h)+i)^{-k}$ is of trace class %for sufficiently large $k\in \mathbb N$. Consequently, $f(P_1(h))-f(P_0(h))$ is a trace class operator for all %$f\in C_0^{\infty}(\mathbb R;\mathbb R)$. We define the spectral shift function $s_{h}(\cdot):=s_{h}%(\cdot;P_1(h),P_0(h))\in \mathcal{D}'(\mathbb R)$ associated to the operators pair $(P_1(h), %P_0(h))$ by the Lifshits-Krein formula
\begin{equation}\label{LK}
\langle s_{h}'(\cdot),f(\cdot) \rangle = - {\rm tr}\big( f(P_1(h))-f(P_0(h)) \big), \quad \forall f\in C_0^{\infty}(\mathbb R;\mathbb R).
\end{equation} 
The function $s_h(\tau)$ is fixed up to an additive constant by the formula \eqref{LK}, and we normalize {it} so that $s_h(\tau)=0$ for $\tau <{\rm inf} (\sigma(P_1(h))$.
%Notice that by this formula, $s_{h}(\cdot)$ is defined modulo a constant, however this is irrelevant since we are concerned with the asymptotic behaviour of the derivative $s_{h}'(\cdot)$ as $h \searrow 0$.

{We denote by $p_1(x,\xi) :=\xi^2 I_N +V(x)$ and $p_0(x,\xi):= \xi^2 I_N + V_{\infty}$, $(x,\xi)\in \mathbb
R^{2n}$, the classical Hamiltonians associated with the operators $P_1(h)$ and $P_0(h)$, respectively.} Let $e_1(x)\leq e_2(x)\leq ... \leq e_N(x)$  be the eigenvalues of $V(x)$  arranged in increasing order.

%In our first result, we give a representation of the SSF in the sense of distribution and a weak asymptotics on %$s'_{h}(\cdot)$ which is an immediate consequence of Theorem \ref{Trace of the difference}. Let $z_0\in %\mathbb R \setminus \text{Spec}(P_j(h))$, $j\in \{0,1\}$, and $q\in \mathbb N$ such that $q>\frac{n}{2}$. %Notice that such $z_0$ can be chosen since $P_0(h)$ and $P_1(\h)$ are bounded from below. For $z\in %\mathbb C\setminus \mathbb R$, the operator $\big[(P_j(h)-z_0)^{-q}(z-P_j(h))^{-1} \big]_0^1$ is of trace %class. We introduce the following function 
%\begin{equation}\label{la fonction sigma}
%\sigma_{h}(z):=(z-z_0)^q {\rm tr}\bigg( \big[(P_j(h)-z_0)^{-q}(z-P_j(h))^{-1} \big]_0^1 \bigg), \quad z\in %\mathbb C \setminus \mathbb R.
%\end{equation} 
\begin{theorem}\label{WA1}
Assume \eqref{DA1} and let $f\in C_0^{\infty}(\mathbb R;\mathbb R)$. There exists a sequence of real numbers $(c_{2j}(f))_{j\in \mathbb N}$  such that % we have the following asymptotics
\begin{equation}\label{Oubl}
\langle s_{h}'(\cdot),f(\cdot) \rangle  \sim (2\pi h)^{-n} \sum_{j\geq 0} c_{2j}(f) h^{2j} \quad \text{as}\;\;h \searrow 0,
\end{equation}
with
\begin{equation}\label{Oubl2}
c_0(f) = \frac{\omega_n}{2} \sum_{k=1}^N   \int_{\mathbb R^n}  \int_0^{+\infty}\big[ f(e_{k,\infty}+ \tau)-f(e_k(x)+\tau)\big] \tau^{\frac{n-2}{2}}d\tau dx,
\end{equation}
{where} $\omega_n$ is the volume of the unit sphere $\mathbb{S}^{n-1}$.
\end{theorem}

For $\tau_0\in \mathbb R$, set
$$\Sigma_{\tau_0}:=\bigcup_{k=1}^N \{(x,\xi)\in {\mathbb R}^{2n}; \xi^2+e_k(x)=\tau_0\}.$$
 The following theorem  is a  consequence of Theorem  \ref{WASG}.

\begin{theorem}[Weak asymptotics]\label{ASW2}
Let $\tau_0\not\in\{e_{1,\infty},e_{2,\infty},\cdots,e_{N,\infty}\}$. Assume \eqref{DA1} and  $\tau_0- {p_1}(x,\xi)$ is microhyperbolic at every point  $(x,\xi)\in \Sigma_{\tau_0}$.
Then, if $\theta$ is equal to 1 near the origin, there exist $I\in O_{\tau_0}$  and $\varepsilon$ small enough and independent of $h$   such that for   $f\in C_0^{\infty}(I;\mathbb R)$,
 the following asymptotic formula holds uniformly {for} $\tau\in\mathbb R$:
\begin{equation} \label{weak dim1}
\langle s_{h}'(\cdot), \mathcal{F}_{h}^{-1}\theta_\varepsilon(\tau- \cdot)f(\cdot)\rangle \sim (2\pi h)^{-n} f(\tau) \sum_{j\geq 0} \gamma_{2j}(\tau)h^{2j} \quad \text{as}\;\;h \searrow 0.
\end{equation}
 The coefficients $\gamma_{2j}(\tau)$ are smooth functions of $\tau$, independent of $f$ and $\theta$. In particular, 
\begin{equation}\label{LTA}
\gamma_0(\tau)= \frac{\omega_n}{2} \sum_{k=1}^N \int_{\mathbb R^n} \big( (\tau-e_k(x))_+^{\frac{n-2}{2}}-(\tau-e_{k,\infty})_+^{\frac{n-2}{2}} \big) dx,
\end{equation}
where  $\tau_+ :=\max\;(\tau,0)$.
\end{theorem}

{\bf Remark 2.3.} 
%Set $\Sigma=\cup_{k=1}^N\{x\in \mathbb R^n; e_k(x)=\tau_0\}.$ 
According to Definition \ref{de1},  the assumption that  $\tau_0- p_1(x,\xi)$ is microhyperbolic at every point  $(x,\xi)\in \Sigma_{\tau_0}$ is equivalent to the following condition: For $x_0$
 with $ e_j(x_0) = \tau_0$, $j=1,...,N$, there exists $T_1\in \mathbb R^n$  and $C>0$ such that
$$\Big(\langle T_1,\nabla_x V(x_0)\rangle \omega, \omega\Big)\geq \frac{1}{C} \vert \omega\vert^2,\,\,\forall \omega \in {\rm ker}(V(x_0)-\tau_0I_N).$$
In particular, if $e_j(x_0)$ is a simple eigenvalue of $V(x_0)$, this is equivalent to $\nabla e_j(x_0)\not =0$.

As a consequence of Theorem \ref{ASW2}, we get a sharp remainder estimate for the  spectral shift  function corresponding to the pair $(P_1(h),P_0(h))$. %For simplicity we assume that $V_\infty=0$.
%Using Theorems \ref{Weak1} and \ref{Weak2}, we prove the following Weyl-type asymptotics on $s_{\h}(\cdot)$.
\begin{theorem}[Weyl-type asymptotics]\label{WRA} {Assume that \eqref{DA1} holds with $V_{\infty}=0$.} Let $\tau_0\not=0$ such that
$\tau_0-p_1(x,\xi)$ is microhyperbolic at every point  $(x,\xi)\in \Sigma_{\tau_0}$.
There exists $I\in O_{\tau_0}$ such that 
\begin{equation}
s_{h}(\tau) = (2\pi h)^{-n} a_0(\tau) + \mathcal{O}(h^{-n+1}) \quad \text{as}\;\;h \searrow 0,
\end{equation} 
uniformly for $\tau\in  I$,
with
\begin{equation}\label{premier terme du developpement}
a_0(\tau)= \frac{\omega_n}{n} \sum_{k=1}^N \int_{\mathbb R^n} \big( (\tau-e_k(x))_+^{\frac{n}{2}}- \tau_+^{\frac{n}{2}} \big) dx.
\end{equation}
%In particular, if $e_{j,\infty}<\tau_0<e_{j+1,\infty}$, then
%\begin{equation}\label{premier terme du d\UTF{00E9}veloppement}
%a_0(\tau)=\sum_{k=j+1}^N{\mathcal N}_k(-\infty,\tau)+ \frac{\omega_n}{n(2\pi)^n} \sum_{k=1}^j\int_{\mathbb R^n} \big( (\tau-e_k(x))_+^{\frac{n}{2}}-(\tau-e_{k,\infty})_+^{\frac{n}{2}} \big) dx,
%\end{equation}
%where ${\mathcal N}_k(-\infty,\tau)$ is the number of eigenvalue of $-h^2\Delta+e_k(x)$ in $]-\infty,\tau]$.

%\textcolor{blue}{Faire le cas $\tau_0<e_{1,\infty}$.}
\end{theorem}
%Let us now recall the following notion of escape function for the symbol $P_1$ at the energy $\lambda_0$.
As indicated in the introduction,  in the scalar case a complete asymptotic expansion in powers of $h$  of the derivative of the SSF  has been
obtained under a non-trapping condition on the classical trajectories corresponding to the energy surface $\Sigma_{\tau_0}$ ({see \cite{Rob}}). In the {present matrix-valued case}, the treatment is much more
complicated. In fact,  since the eigenvalues are not enough regular, the usual definition of the Hamilton flow  for a matrix-valued Hamiltonian function  does not make sense (see \cite{Kato}). For this reason,   we use here the notion of escape function. 

More precisely, we suppose that there exists a scalar {escape} function $G\in C^{\infty}(\mathbb R^{2n};\mathbb R)$ associated to $p_1$ at $\tau_0$, i.e.,
\begin{equation}\label{deff} 
\exists \: C>0, \text { s.t. } \{p_1,G\}(x,\xi) :=  \frac{\partial G}{\partial x} \cdot\frac{\partial p_1}{\partial \xi}  - \frac{\partial G}{\partial \xi} \cdot \frac{\partial p_1}{\partial x}  \geq C  , \quad \forall (x,\xi)\in \Sigma_{\tau_0},
\end{equation}
in the sense of hermitian matrices. 

In the scalar case $N=1$, it is well known that the above assumption is equivalent to the non-trapping condition on the energy $\tau_0$. In fact, if $\tau_0$ is non-trapping for the classical Hamiltonian $p_1$, one can construct an escape function $G\in C^{\infty}(\mathbb R^{2n};\mathbb R)$ satisfying \eqref{deff} (see for instance \cite{GM}, \cite{Wa1}, \cite{Wa2}, \cite{Wa3}). Conversely, if \eqref{deff} holds then one easily sees that $G$ is strictly increasing along the Hamiltonian flows associated to $p_1$ in $\Sigma_{\tau_0}$ which prevents the existence of trapped trajectories at $\tau_0$. We also point out that \eqref{deff} implies that $\tau_0-p_1(x,\xi)$ is microhyperbolic at every point $(x,\xi)\in \Sigma_{\tau_0}$ in the direction of the Hamiltonian vector field $(\partial_{\xi}G(x,\xi),-\partial_{x} G(x,\xi))$.

Now we can formulate the main result of this paper.

\begin{theorem}[Strong asymptotics]\label{main semiclassical}
Fix an energy $\tau_0>e_{N,\infty}$. Assume that \eqref{DA1} and \eqref{deff} are satisfied. Then, there exists $I\in O_{\tau_0}$  such that $s_{h}'(\cdot)$ has a complete asymptotic expansion of the form
\begin{equation}\label{main semiclassical asymptotics}
s_{h}'(\tau) \sim {(2\pi h)^{-n}} \sum_{j\geq 0} \gamma_{2j}(\tau)h^{2j} \quad \text{as}\;\;h \searrow 0,
\end{equation}
uniformly for $\tau\in I$, where the coefficients $\gamma_{2j}(\tau)$ are given in Theorem \ref{ASW2}.
\end{theorem}

\subsection{Examples and further generalizations} 
%In this paragraph, we give some examples of hermitian matrix-valued potentials satisfying the assumption \eqref{deff} with 
%$$G(x,\xi) = x \cdot \xi, \quad (x,\xi)\in \mathbb R^{2n}.$$

First  observe that, for $G(x,\xi)=x\cdot \xi$,  \eqref{deff} is equivalent to
\begin{equation}\label{generateur des dilatations}
2(\tau_0 -e_k(x))-x\cdot \nabla V(x) \geq C,\,\,\,\,\, \forall x\in \{x\in \mathbb R^n;\: \tau_0 - e_k(x) \geq 0 \}, k=1,\cdots N.
\end{equation}
Thus, under the assumption  \eqref{DA1}, the asymptotics  \eqref{main semiclassical asymptotics}  holds near any  large $\tau_0$ with
 $$\tau_0>{\rm sup}_{{x\in \mathbb R}^n} \Vert \frac{x\cdot \nabla V(x)}{2}\Vert_{N\times N}+{\rm sup}_{x\in {\mathbb R}^n}\Vert V(x)\Vert_{N\times N}.$$
% On the other hand, if $
%x \cdot \nabla V(x) <0$ for all  $x\in \mathbb R^n,$ then   \eqref{main semiclassical asymptotics}   holds for all
 %$\tau_0>e_{N,\infty}$.
%\item[(3)]

 Notice that our results extend to the case of potentials depending on $h$, i.e. $V(x;h) = V_0(x) + h V_1(x;h)$. In such a case, we assume \eqref{DA1} uniformly with respect to $h$.
In particular, as a simple example, consider the case where $V_0(x)$ is  a diagonal matrix ${\rm diag}\, (e_1(x),\ldots,e_N(x))$.
If each $e_j(x)$ satisfies
$$
2 (\tau_0 -e_j(x)) - x \cdot\nabla_x e_j(x) \geq c_j> 0 , \quad  \forall x\in \{x\in \mathbb R^n;\: \tau_0 - e_j(x) \geq 0 \},
$$
then \eqref{generateur des dilatations} is satisfied for $h$ small enough and  \eqref{main semiclassical asymptotics}   holds.

More generally, we can treat the spectral shift function associated to a pair of self-adjoint $h$-pseudodifferential operators $(P_1(h),P_0(h))$ provided that the SSF is well defined and the existence of a scalar escape function holds. 
\subsection{Outline of the proofs}

The purpose of this subsection is to provide an outline of the proofs. % of Theorems \ref{th1}, \ref{th2}, \ref{WASG} and \ref{main semiclassical}.

%\subsubsection{Semiclassical trace formula} 
As indicated in the introduction, our method is time-independent. The starting point is the functional calculus of $h$-pseudodifferential operators based on the Helffer-Sj\"ostrand formula (see \cite[Ch. 8]{dimassi}). By this formula, the main object to study  will be the integral of the form
\begin{equation}\label{basic expression}
{\mathcal I}(\tau,\varepsilon;h) =  - \frac{1}{\pi} \int_{\mathbb C} \bar{\partial} \tilde{f}(z) \mathcal{F}_{h}^{-1}\theta_{\varepsilon} (\tau-z) K(z;h) L(dz), \quad \tau \in \mathbb R,
\end{equation}
where $L(dz)= dxdy$ is the Lebesgue measure on $\mathbb C \sim \mathbb R^2$. 
Here, $\tilde{f} \in C_0^{\infty}(\mathbb C)$ denotes an almost analytic extension of $f\in C_0^{\infty}(\mathbb R; \mathbb R)$ (see \cite[Ch. 8]{dimassi} and also \cite{Hormander}), i.e., 
\begin{equation}\label{prop1}
\tilde{f}_{| \mathbb R} = f,
\end{equation}
\begin{equation}\label{prop2}
\bar{\partial} \tilde{f}(z) = \mathcal{O}(|\Im z|^{\infty}),
\end{equation}
and $K$, which in fact is the trace of an operator depending on the resolvent,  is a complex-valued analytic function defined  in a neighborhood of ${\rm supp}\: \tilde{f}$ except on the real axis,
 with an estimate
\begin{equation}\label{estimate G}
K(z;h) = \mathcal{O}\left( h^{-n} |\Im z|^{-2} \right).
\end{equation}
The right hand side of \eqref{basic expression} is independent of the particular choice of the almost analytic extension $\tilde{f}$. In particular,
let $\psi_L(z)$ be a function on $\mathbb C$ defined by
\begin{equation}
\label{L}
\psi_L (z) = \psi (\frac{\Im z}{L}),\quad L>0,\quad C_0^{\infty}(\mathbb R; \mathbb R)\ni \psi(t)=
\left\{
\begin{array}{l}
1\quad (|t|\leq 1) \\[8pt]
 0\quad (|t|\geq 2).
\end{array}
\right.
\end{equation}
Then $\tilde{f}\psi_L$ is also an almost analytic extension of $f$, and we have
\begin{equation}\label{basic expression1}
{\mathcal I}(\tau,\varepsilon;h) =  - \frac{1}{\pi} \int_{\mathbb C} \bar{\partial}\left(\tilde{f}\psi_L\right)(z) \mathcal{F}_{h}^{-1}\theta_{\varepsilon} (\tau-z) K(z;h) L(dz).
\end{equation}
From now on, $M>0$ is a constant independent of $h$ and we put
\begin{equation}\label{zeta}
\zeta(h) := h \log(\frac{1}{h}), \quad L := \frac{M \zeta(h)}{\varepsilon}.
\end{equation}

We begin with a general remark  on the integral given by the right hand side of  \eqref{basic expression1}.
From  \eqref{prop2} and the definition of $\psi_L$, we deduce 
\begin{equation}\label{til}
\bar{\partial}\big( \tilde{f} \psi_L\big)(z) = 
\mathcal{O}(h^{\infty}) \psi_L(z) + \mathcal{O}(\frac{1}{L} ) \tilde{f}(z)1_{[1,2]\cup [-2,-1]}\big(  \frac{ \Im z}{L} \big),
\end{equation}
 which together with \eqref{-11} yields ${\mathcal I}(\tau,\varepsilon;h)\equiv {\mathcal  I}_+(\tau,\varepsilon;h) +{\mathcal  I}_-(\tau,\varepsilon;h)$, uniformly for $0<\varepsilon \leq c h^{-\nu}$ (where $\nu$ is a fixed constant). Here
\begin{equation}\label{D1111}
{\mathcal I}_\pm(\tau,\varepsilon;h):=-\frac{1}{\pi}\int_{ \{\pm\Im z>L\}} \bar{\partial}\left( \tilde{f} \psi_L\right)(z) \mathcal{F}_{h}^{-1}\theta_{\varepsilon}(\tau-z) \, K(z;h)L(dz) .
\end{equation}
We recall that the notation $A\equiv B$ means $A-B={\mathcal O}(h^\infty)$.
The behavior of the function $\mathcal{F}_{h}^{-1}\theta_{\varepsilon}(\tau-z)$ depends on the support of $\theta$. For general $\theta$ with support in $]-1,1[$, we have
\begin{equation}
\label{-11}
\mathcal{F}_{h}^{-1}\theta_{\varepsilon}(\tau-z)={\mathcal O}(\frac \varepsilon h e^{\frac{\varepsilon |\Im z|}h}). 
\end{equation}
In particular, in the support of $\psi_L$, we have
\begin{equation}\label{nouvelle formule}
\mathcal{F}_{h}^{-1}\theta_{\varepsilon}(\tau-z)={\mathcal O}(\varepsilon h^{-2M-1}). 
\end{equation}
For $ \theta$ with support only in ${\mathbb R}_+$, say  in $]\frac{1}{2},1[$, we have  
\begin{equation}\label{PW}
\mathcal{F}_{h}^{-1}\theta_{\varepsilon}(\tau-z) = \left\{ \begin{array}{rcl}
 \mathcal{O}\big(\frac{\varepsilon}{h} e^{\frac{\varepsilon \Im z}h} \big), \quad \Im z >0, \\[8pt]
 \mathcal{O}\big(\frac{\varepsilon}{h} e^{\frac{\varepsilon \Im z}{2h}} \big), \quad \Im z < 0.
\end{array}\right.
\end{equation}
This latter estimate implies in particular that
\begin{equation}\label{D1113}
{\mathcal I}_-(\tau,\varepsilon;h)={\mathcal O}(\varepsilon^2 h^{\frac{M}{2}-n-2}),
\end{equation}
which means ${\mathcal I}_-={\mathcal O}(h^\infty)$ if $\varepsilon$ is at most of polynomial order in $h$ and $M$ is arbitrary.

Let $\theta,\varepsilon$ be as in Theorem \ref{th1} and assume that  $\tau_0 - H(x,\xi)$ is uniformly microhyperbolic on ${\mathbb R}^{2n}$ in  the direction $T$. Without any loss of generality we may assume that $\theta\in C_0^{\infty}(]\frac{1}{2},1[;\mathbb R)$. According to  the  Helffer-Sj\"ostrand formula  (see  \eqref{formulaM}, \eqref{HS}), the left hand side of \eqref{main estimate 1} can be written as \eqref{basic expression}  with 
$$K(z;h):= (\lambda_0-z)^{k-1} {\rm tr}\big( A^w (z-H^w)^{-1}(\lambda_0-H^w)^{-(k-1)} \big),$$ 
where $\lambda_0 < \inf(\sigma(H^w))$ and $k \in \mathbb N$ is large enough so that $A^w (z-H^w)^{-1}(\lambda_0-H^w)^{-(k-1)}$ is of trace class. 

As explained above, we have ${\mathcal I}_-={\mathcal O}(h^\infty)$. To deal with ${\mathcal I}_+$, we conjugate the operator $A^w (z-H^w)^{-1}(\lambda_0-H^w)^{-(k-1)}$ with the unitary operator $U_t := e^{\frac{it}{h}(T_2\cdot x - T_1\cdot h D_x)}$, $t\in \mathbb R$.
Here $T=(T_1,T_2)$ is the direction of the uniform microhyperbolicity of $\tau_0-H(x,\xi)$. Then the function
$$
K_{t}(z;h) := 
(\lambda_0-z)^{k-1}  {\rm tr}\big( A_{t}^w (z- H_{t}^w)^{-1} (\lambda_0-H_{t}^w)^{-(k-1)} \big) , 
$$
with 
$H_t^w := U_t H^w(x, hD_x) U_{t}^{-1}=H^w(x+tT_1, hD_x+tT_2)$ etc.,
is invariant with respect to the change of  real $t$ and coincides with  $K(z;h)$ thanks to the cyclicity of the trace.

Now, replacing $H,A$ by their almost analytic extensions ($\tilde{H}$ and $\tilde{A}$), we extend this function to complex $t$. The extended function $\tilde K_t(z;h)$ is defined in $\{z\in \mathbb C; \, \Im z>C_0\Im t\}$ for some positive constant $C_0$  independent of $M$,$\varepsilon$ and $h$.
We fix $t_0=iL/C_0$. Then we see that $\tilde K_{t_0}(z;h)$ is equal to $K(z;h)$ modulo ${\mathcal O}\,(h^\infty)$ in the domain $\{ \Im z>L\}$.

The uniform microhyperbolic condition enables us to continue $\tilde K_{t_0}(z;h)$ analytically to the lower half plane with $\Im z>-cL$ for a positive contant $c$.
In fact, the imaginary part of the Weyl symbol of $z-\tilde H_{t_0}$  stays positive definite in such a region, and the sharp G\aa{}rding inequality
guarantees the invertibility of the operator.

In the integral expression \eqref{D1111} of ${\mathcal I}_+$, we can replace, modulo ${\mathcal O}\,(h^\infty)$, $K(z;h)$ by $\tilde K_{t_0}(z;h)$ and then the integral domain
$\Im z>L$ by $\Im z <-cL$ by the Cauchy theorem. Thus the estimate ${\mathcal I}_+={\mathcal O}(h^\infty)$  is reduced to the same argument as for ${\mathcal I}_-$, and 
we conclude ${\mathcal I}={\mathcal O}(h^\infty)$. This gives Theorem \ref{th1}.

Let us now outline the proof of Theorem \ref{th2}. By  Helffer-Sj\"ostrand formula, the left hand side of  \eqref{man2}  can be written as  \eqref{basic expression} with $K(z;h) = {\rm tr}\,( \chi^w [  (z-H_j^w)^{-1} ]_0^1)$.  Using that
${\rm dist} ({\rm supp}\,\chi, {\rm supp}\, (H_0- H_1) )>0$,  we prove by some exponentially weighted resolvent estimates
%it follows from   the first resolvent identity that
$$
K(z;h) = \mathcal{O}\left(\frac{h^{\frac{M}{\varepsilon C}- n}}{L^2} \right),
$$
uniformly for $|\Im z|\geq L$, where $C>0$ is a constant independent of $h,M,\varepsilon$. 
Combining this with \eqref{nouvelle formule}, we get Theorem \ref{th2}
provided that $\varepsilon>0$ is small enough.

Theorem \ref{WASG} is a consequence of the two previous theorems and the symbolic calculus of $h$-pseudodifferential operators. Assuming that $\chi$ is supported in a small neighborhood of a fixed point $(x_0,\xi_0)\in \mathbb R^{2n}$ (by a partition of unity there is no loss of generality in doing so) and using the fact that changing $H$ outside the support of $\chi$ leads to an error of order $\mathcal{O}(h^{\infty})$ in the trace formula ${\rm tr} (\chi^w f(H^w) \mathcal{F}_h^{-1}\theta(\tau-H^w))$ (according to Theorem \ref{th2}), together with Theorem \ref{main result appendix}, we may assume that there exists $I\in O_{\tau_0}$ such that $\tau - H$ is uniformly microhyperbolic with respect to $(x,\xi)\in \mathbb R^{2n}$ and $\tau \in I$.

Now, fix $\tilde \varepsilon=h^{1-\delta_0}$ with $\delta_0\in ]0,1/2[$.  Applying Theorem  \ref{th1}, we obtain
\begin{equation}\label{man13}
{\rm tr}\; \left(  \chi^w  f (H^w) \mathcal{F}_{h}^{-1}\theta (\tau-H^w)  \right) \equiv {\rm tr}\; \left(  \chi^w  f (H^w) \mathcal{F}_{h}^{-1}\theta_{\tilde\varepsilon} (\tau-H^w)  \right).
\end{equation}
In fact we can represent the difference $\theta-\theta_{\tilde\varepsilon}$ as a finite sum of functions $\tilde{\theta}_\varepsilon$ appearing in Theorem \ref{th1}
(with  $\varepsilon\in [\tilde\varepsilon, \frac{1}{C}[$). 
The principal significance of  \eqref{man13} is that  it allows one to use the  standard $h$-pseudodifferential calculus  and get the asymptotic expansion in powers of $h$ given in Theorem \ref{WASG} just by symbolic calculus (see \cite[Ch. 7-8]{dimassi}) . 
%%%%%%%%%%%%%%%%%%%%%%%%%%%%%%%%%%%%%%%%%%%%%%%%%%%%%%%%%%%%%%%%%%%%%%%%%%%%%%%%%%%%%%%%%%%%%%%%%%%%%%%%%%%%%%%%%%%%%%%%%%%%%%%%%%%% Proofs of the results of the semiclassical trace formula %%%%%%%%%%%%%%%%%%%%%%%%%%%%%%%%%%%%%%%%%%%%%%%%%%%%%%%%%%%%%%%%%%%%%%%%%%%%%%%%%%%%%%%%%%%%%%%%%%%%%%%%%%%%%%%%%
To see  this, we  first recall that for $\vert \Im z\vert>h^\delta$ (with $\delta\in ]0,1/2[$) the resolvent $(z-H^w)^{-1}$ is an $h$-pseudodifferential operator and its corresponding symbol admits  an asymptotic expansion in powers of $h$ (see \eqref{Reso}). Combining this with the fact 
$${\rm tr}\; \left(  \chi^w  f (H^w) \mathcal{F}_{h}^{-1}\theta_{\tilde\varepsilon} (\tau-H^w)  \right)$$
$$\equiv - \frac{1}{\pi} \int_{\{\vert \Im z\vert>h^{\delta_0}\}}   \bar{\partial}  \left(       \tilde f    \psi(\frac{\Im z}{h^{\delta_0}}) \right)(z) \mathcal{F}_{h}^{-1}\theta_{\tilde\varepsilon}(\tau-z) {\rm tr}\left( \chi^w (z-H^w)^{-1} \right) L(dz),
$$
we see that the left hand side of \eqref{man13} has a complete asymptotic expansion in powers of $h$, which yields Theorem \ref{WASG}.

Turn now to the main ideas in the proofs of the results of subsection \ref{SSF} concerning our application to the SSF. Theorem \ref{WA1} is a simple consequence of the $h$-pseudodifferential symbolic calculus while Theorems \ref{ASW2} and \ref{WRA} are consequences of Theorem \ref{WASG} and standard Tauberian arguments combined with a trick of Robert \cite{Robert0} respectively.

Finally, we sketch the proof of our main result which is Theorem \ref{main semiclassical}. According to \eqref{LK} and the Helffer-Sj\"ostrand  formula we have
$$- \langle s_{h}'(\cdot), \mathcal{F}_{h}^{-1}\theta_{\varepsilon}(\tau- \cdot)f(\cdot)\rangle={\mathcal I}(\tau,\varepsilon;h),$$
with
$$
K(z;h)=(z-\lambda_0)^q {\rm tr}\bigg( \big[(P_{j}(h)-\lambda_0)^{-q}(P_{j}(h)-z)^{-1} \big]_0^1 \bigg).
$$
Here $\lambda_0<\inf \sigma(P_1(h))$ and $q \in \mathbb N$ is large enough, see \eqref{weak dim}. 

First, suppose that $0$ is not contained in the support of $\theta$. Then, ${\mathcal I}_-(\tau,\varepsilon;h)={\mathcal O}(h^\infty)$ uniformly for $\varepsilon\in ]h^{1-\delta}, h^{-\nu}[$ as before.

To deal with ${\mathcal I}_+(\tau,\varepsilon;h) $, we adapt an idea from the theory of resonance. More precisely, 
under the existence of an escape function near $\tau_0$ (assumption   \eqref{deff}), we prove by  the analytic distortion method that $K(z;h)$ extends analytically from the upper half plane to the lower one with $ \Im z>-M\zeta(h)$ for all $M>0$. 

From this, we deduce two important consequences. First,
\begin{equation}\label{Oubl100001}
s_{h}^{''}(\tau) = \mathcal{O}\big(h^{-n}\zeta(h)^{-2}\big).
\end{equation}
uniformly for $\tau$ near $\tau_0$.
Second, the same argument as  in the proof of  Theorem \ref{th1} leads to
$ {\mathcal I}_+(\tau,\varepsilon;h) ={\mathcal O}(h ^{\infty}),$
uniformly for $\varepsilon\in ]h^{1-\delta}, h^{-\nu}[$. Hence we obtain
\begin{equation}\label{Oubl10000}
\langle s_{h}'(\cdot), \mathcal{F}_{h}^{-1}\theta_{\varepsilon}(\tau- \cdot)f(\cdot)\rangle={\mathcal O}(h ^{\infty}).
\end{equation}

Now, we assume that $\theta$ is equal to one near zero, and let $\varepsilon$ be small and independent of $h$ and $\tilde\varepsilon=h^{-\nu}$. 
As in the proof of  \eqref{man13},  the formula \eqref{Oubl10000}  yields 
$$\langle s_{h}'(\cdot), \mathcal{F}_{h}^{-1}\theta_{\varepsilon}(\tau- \cdot)f(\cdot)\rangle \equiv \langle s_{h}'(\cdot), \mathcal{F}_{h}^{-1}\theta_{\tilde\varepsilon}(\tau- \cdot)f(\cdot)\rangle.$$
By \eqref{LK} and   \eqref{weak dim1}, the left hand side of the above equality has an asymptotic expansion in powers of $h$.  
%On the other hand, 
%Now the proof is complete by observing that 
The right hand  side is written,  by Taylor's formula  and \eqref{Oubl100001}, 
$$\langle s_{h}'(\cdot), \mathcal{F}_{h}^{-1}\theta_{\varepsilon}(\tau- \cdot)f(\cdot)\rangle= s_{h}'(\tau)f(\tau)+ \mathcal{O}(h^{\nu+1-n}\zeta(h)^{-2}).$$
Since $\nu$ is arbitrary, this ends the proof of Theorem \ref{main semiclassical} by taking $f=1$ near $\tau_0$.

\section{Proofs of the results on the semiclassical trace formula}\label{Proofs of the results STF}
In this section, we prove the results concerning the semiclassical trace formula. Throughout our proofs, when it is not precised, we let $C$ denotes a positive constant that may take different values, but is always independent of $\varepsilon, h$ and $M$.

\subsection{Proof of Theorem \ref{th1}}{
Writing $\theta = \theta_1+\theta_2$, with ${\rm supp}\: \theta_1 \subset]0,+\infty[$ and ${\rm supp}\: \theta_2 \subset ]-\infty,0[$, we may assume that ${\rm supp}\: \theta \subset ]\frac{1}{2},1[$.

%In the following $C_i$ denotes a large constant (independent on $M,h$) which  changes  from line to line.  

  For $\tau\in \mathbb R$ and $\varepsilon>0$, we define 
\begin{equation}\label{TrFor}
{\mathcal I}(\tau,\varepsilon;h) := {\rm tr}\; \left(  A^w f\big(H^w\big) \mathcal{F}_{h}^{-1}\theta_{\varepsilon}\big(\tau-H^w\big) \right).
\end{equation}
Let $\tilde{f}$ be an almost analytic extension of $f$ satisfying \eqref{prop1} and \eqref{prop2} with ${\rm supp} \tilde{f} \subset \{z\in \mathbb C;\: |\Im z|\leq 1\}$.
If $g$ is real analytic in a neighborhood of the support of $\tilde f$, then we have, by the Helffer-Sj\"ostrand formula (see \cite[Ch. 8]{dimassi}),  
\begin{equation}\label{HeSj}
f(H^w)g(H^w)= - \frac{1}{\pi} \int_{\mathbb C} \bar{\partial} \tilde{f}(z) \, g(z) \,  (z-H^w)^{-1}L(dz). 
\end{equation}

Let $\lambda_0\in\mathbb R$ be fixed such that $\lambda_0<{\rm inf}(\sigma(H^w))$ and set, for $ \Im z\not=0$,
\begin{equation} \label{formulaM}
K(z;h) := (\lambda_0-z)^{k-1} {\rm tr}\big( A^w (z-H^w)^{-1}(\lambda_0-H^w)^{-(k-1)} \big).
\end{equation}
Then,
using \eqref{TrFor}, \eqref{HeSj}  and \eqref{formulaM} with $g(z) = (\lambda_0-z)^{k-1}\mathcal{F}_{h}^{-1}\theta_{\varepsilon}(\tau-z)$, we obtain
\begin{equation}\label{HS}
{\mathcal I}(\tau,\varepsilon;h) = - \frac{1}{\pi} \int_{\mathbb C} \bar{\partial} \tilde{f}(z) \mathcal{F}_{h}^{-1}\theta_{\varepsilon}(\tau-z)  \, K(z;h) L(dz). 
\end{equation}
Let $L$ and $\psi_L$ be defined by \eqref{zeta} and \eqref {L}. We write
\begin{equation}\label{D3}
{\mathcal I}= {\mathcal I}_+ +{\mathcal  I}_-,\quad
{\mathcal I}_{\pm}:= - \frac{1}{\pi} \int_{\{\pm \Im z>0\}} \bar{\partial}( \tilde{f}\psi_L)(z) \mathcal{F}_{h}^{-1}\theta_{\varepsilon}(\tau-z) \, K(z;h) L(dz).
\end{equation}
Since the support of $\theta$ is included in $]\frac{1}{2},1[$, it follows from \eqref{D1113} that 
\begin{equation}
\label{i-}
{\mathcal I}_-(\tau,\varepsilon;h) \equiv 0,
\end{equation}
uniformly for $\tau\in \mathbb R$ and $\varepsilon\in [h^{1-\delta},\kappa[$, for all $\kappa>0$.

Let us now turn to the study of ${\mathcal I}_{+}(\tau,\varepsilon;h)$. By assumption, there exists $T=(T_1,T_2)\in \mathbb R^{2n}$ and $I_{\tau_0}\in O_{\tau_0}$ such that $\tau - H(x,\xi)$ is uniformly microhyperbolic in the direction $T$ with respect to $(x,\xi) \in \mathbb R^{2n}$ and $\tau \in I_{\tau_0}$.

For $t\in \mathbb R$, we define  the unitary operator
$$
U_t:= e^{\frac{it}{h}(T_2\cdot x-T_1\cdot hD_x)} .
$$
%where $A(x,\xi) := \langle T, (x,\xi)\rangle = \langle T_1,x \rangle + \langle T_2, \xi \rangle$. 
Clearly, we have 
\begin{equation*}
H_t^w := U_t^{-1} H^w(x,h D_x) U_t = H^w\big((x,h D_x)+tT\big)= H^w\big(x+tT_1,h D_x+tT_2\big),
\end{equation*}
\begin{equation*}
A_t^w := U_t^{-1} A^w(x,h D_x) U_t =  A^w\big((x,h D_x)+tT \big) = A^w\big(x+tT_1,h D_x+tT_2\big).
\end{equation*}
Let $\tilde{H},\tilde{A}$ be two almost analytic extensions of $H$ and $A$, respectively, which are bounded together with all theirs derivatives. Put for complex $t$ with small imaginary part
$$
\tilde{H}_t^w:= \tilde{H}^{w}((x,h D_x)+t T) \quad \text{and} \quad \tilde{A}_t^w:= \tilde{A}^{w}((x,h D_x)+t T).
$$
By Taylor's formula with respect to $\Im t$, we have 
\begin{eqnarray}\label{Cons Taylor}
z- \tilde{H}((x,\xi)+t T)  &=& z- \tilde{H}\big((x,\xi) + \Re t T + i \Im t T \big) \nonumber\\
& = & z- H\big((x,\xi) + \Re t T\big) - i \Im t \langle \nabla_{x,\xi} H((x,\xi)+\Re t T), T \rangle + \mathcal{O}(|\Im t|^2).
\end{eqnarray}
Thus, one easily sees by using the Calder\'on-Vaillancourt theorem (see \cite[Theorem 7.11]{dimassi}) that there exists a constant $C_0>0$ (depending only on the $L^{\infty}$-norms of a finite numbers of derivatives of $H$) such that $(z-\tilde{H}_t^w)^{-1}$ exists for $|\Im z|\geq C_0 |\Im t|$. Set 
$$
 \tilde{K}_t(z;h):=(\lambda_0-z)^{k-1}  {\rm tr}\big( \tilde{A}_t^w (z- \tilde{H}_t^w)^{-1} (\lambda_0-\tilde H_t^w)^{-(k-1)} \big).
$$
Using that
$\overline{\partial_t} \tilde{A}_t$, $\overline{\partial_t} \tilde{H}_t={\mathcal O}(\vert \Im t\vert^\infty)$, we obtain, uniformly on $\{z\in \mathbb C;\; |\Im z|\geq C_0 |\Im t|\}$,
\begin{equation}\label{estima}
\bar{\partial}_t \tilde{K}_t(z;h) = \mathcal{O}\bigg( \frac{|\Im t|^{\infty}}{|\Im z|^2} \bigg).
\end{equation}

On the other hand, since $U_t$ is unitary for $t\in \mathbb R$, it follows from the cyclicity of the trace that $\tilde{K}_t$ is independent of  $\Re t$. This implies 
\begin{equation}
\bar{\partial}_t \tilde{K}_t(z;h) = \frac{i}{2} \partial_{\Im t} \tilde{K}_t(z;h), \quad {\rm and}\quad
\tilde{K}_t(z;h) = K(z;h) , \quad \forall t\in \mathbb R.
\end{equation}
We have, uniformly for $|\Im z| \geq C_0 |\Im t|$, 
%\begin{eqnarray}
$$K(z;h) - \tilde{K}_{i \Im t}(z;h) = \tilde{K}_{\Re t}(z;h) - \tilde{K}_{\Re t + i\Im t}(z;h)= -\int_0^{\Im t} \frac{d}{ds}\tilde{K}_{\Re t + i s}(z;h) ds =\mathcal{O}\bigg( \frac{|\Im t|^{\infty}}{|\Im z|^2} \bigg).
$$

Fix $t_0 = iL/C_0$. By the preceding estimate, we have, uniformly for $|\Im z| \geq  L$,
\begin{equation}\label{iuy}
K(z;h) - \tilde{K}_{t_0}(z;h) = \mathcal{O}(h^{\infty}).
\end{equation}
In the expression \eqref{D3} of $I_+$,
 one sees from (\ref{til})  and (\ref{PW})    that the restriction of the integral to the domain $0<\Im z \leq L$ is $\mathcal{O}(h^{\infty})$. Therefore, by (\ref{iuy}), we get 
\begin{eqnarray*}
{\mathcal I}_+(\tau,\varepsilon;h) &\equiv & - \frac{1}{\pi} \int_{ \{\Im z> L\}} \bar{\partial}\big( \tilde{f}\psi_L\big)(z) \mathcal{F}_{h}^{-1}\theta_{\varepsilon}(\tau-z) K(z;h) L(dz) \\
& \equiv  & - \frac{1}{\pi} \int_{ \{ \Im z>L \}} \bar{\partial}\big( \tilde{f}\psi_L\big)(z) \mathcal{F}_{h}^{-1}\theta_{\varepsilon}(\tau-z) \tilde{K}_{t_0}(z;h) L(dz).
\end{eqnarray*}
%\textcolor{blue}{uniformly for $\tau \in \mathbb R$.}
\begin{lemma}
Let $t _0= iL/C_0=\frac{iM}{C_0 \varepsilon}\zeta(h)$  be as above. The function $z\mapsto
 \tilde{K}_{t_0}(z;h) $
extends as a holomorphic function to the zone $\Im z\ge - {|t_0|\over 2}$.
\end{lemma}

\begin{proof}
{As in \eqref{Cons Taylor}, Taylor's formula yields} 
$$
z-\tilde H_{t_0}(x,\xi )=z-H(x,\xi
)- t_0\langle T,\nabla_{x,\xi} H(x,\xi )\rangle +{\mathcal O}( |t_0|^2).
$$ 
%Put $\mu =\Imt={M\over C_1\varepsilon }\zeta(h)$.
%and let $-{\mu \over 2}<\Im z<0$. 
Using the
global  microhyperbolicity condition,
 we obtain for %$\Re z\in ]\tau_0-\eta,\tau_0+\eta[$ (with $\eta$ 
 small $h$
%$$\Im (z-\tilde H_{i\mu} (x,\xi ))\geq (C\mu +\Im z)I-{\mathcal O}(1)\mu (\Re z-H(x,\xi ))^2$$ in the
%sense of Hermitian matrices. Here $C$ is fixed constant given by the microhyperbolicity condition. Then,
%if $h$ is small enough, there exists a constant
%$C_2>0$ such that,
\begin{equation}\label{Garding11}
\Im (z-\tilde H_{t_0} (x,\xi ))+C|t_0|
(z-\tilde H_{t_0} (x,\xi ))^*(z-\tilde H_{t_0} (x,\xi ))\geq
c(|t_0| +\Im z)I_N,
\end{equation}
uniformly on $z$ with $\Im z>0$ and $\Re z\in I$ (see \eqref{Garding1}), {where $C,c>0$ are constants independent of $h$ and $M$.}  Here,  $*$ stands for the usual complex adjoint of a matrix. 

Now we pass from the symbolic calculus  level to the  $h$-pseudodifferential calculus.
 The semiclassical version of the sharp G\aa{}rding
inequality (see \cite{dimassi} Theorem 7.12 and \cite[Ch.1]{iv}  for the matrix case)  and \eqref{Garding11} imply,
\begin{equation}\label{GI}
\Im ({\rm Op}_h^w(z-\tilde H_{t_0}
)u, u)+C|t_0|\Vert {\rm Op}_h^w(z-\tilde H_{t_0}
)u\Vert ^2 
\end{equation}
$$
 \geq c(|t_0|+\Im z)\Vert u\Vert ^2-{\mathcal
O}(h)\Vert u\Vert ^2\ge {c\over 3} (|t_0| +\Im z)
\Vert u\Vert ^2,
$$
for all $u\in L^2(\mathbb R^n; \mathbb C^N)$ and $h$ small enough. Here we used the fact that $h=o(|t_0|)$. Combining \eqref{GI}
with the inequality $ab\le {c|t_0|\over
6}a^2+{3\over 2 c|t_0| }b^2$, we obtain
$$   {c \over 3}(|t_0| +\Im z)\Vert u\Vert ^2   \leq
\Vert {\rm Op}_h^w(z- \tilde{H}_{t_0} )u\Vert \Vert
u\Vert +C|t_0|\Vert {\rm Op}_h^w(z-\tilde{H}_{t_0}
)u\Vert ^2  
$$
$$ 
 \leq{c |t_0|
\over 6}\Vert u\Vert ^2+({3\over 2 c |t_0| }+C |t_0|
)\Vert {\rm Op}_h^w(z-\tilde{H}_{t_0} )u\Vert ^2,$$
which yields
%Then there exists C$_2>0$ such that,
\begin{equation}
{c\over 6} (|t_0|+\Im z)\Vert u\Vert ^2\leq ({3\over 2 c |t_0| }+C |t_0|
)\
\Vert {\rm Op}_h^w(z-\tilde{H}_{t_0} )u\Vert ^2,
\end{equation}
for all
$u\in L^2({\mathbb R}^n; \mathbb C^N).$ We conclude
that $(z-\tilde{H}_{t_0} ^w)^{-1}  $ extends as a holomorphic function
of $z$ to the zone $\Im z\geq -\frac{|t_0|}{2}$.
This ends the proof of the lemma.
\end{proof}

Let $\tilde{\psi}\in C^{\infty}(\mathbb R;\mathbb R)$ be such that 
$
\tilde{\psi}(s) = \psi(s) $ for $s>0$,
$
\tilde{\psi}(s)=1$ for $ - 1/4C_0 < s < 0,
$ and
$
\tilde{\psi}(s) = 0 $ for $ s< - 1/2C_0,
$
and define $\tilde\psi_L$ as in \eqref{L}. Then
we have 
\begin{eqnarray}
{\mathcal I}_+(\tau,\varepsilon;h) & \equiv  & - \frac{1}{\pi} \int_{ \{\Im z>L\}} \bar{\partial}\big( \tilde{f}\psi_L\big)(z) \mathcal{F}_{h}^{-1}\theta_{\varepsilon}(\tau-z)  \tilde{K}_{t_0}(z;h) L(dz) \nonumber \\
& \equiv &  - \frac{1}{\pi} \int_{ \{\Im z>0\}} \bar{\partial}\big( \tilde{f} \psi_L \big)(z) \mathcal{F}_{h}^{-1}\theta_{\varepsilon}(\tau-z) \tilde{K}_{t_0}(z;h) L(dz) \nonumber  \\
&\equiv & -\frac{1}{\pi} \int_{ \{\Im z>0\} }  \bar{\partial}\big( \tilde{f}\psi_L \tilde{\psi}_L\big)(z) \mathcal{F}_{h}^{-1}\theta_{\varepsilon}(\tau-z) \tilde{K}_{t_0}(z;h) L(dz)  \nonumber \\
&\equiv & \frac{1}{\pi} \int_{ \{\Im z < 0\} }  \bar{\partial}\big( \tilde{f}\psi_L \tilde{\psi}_L\big)(z) \mathcal{F}_{h}^{-1}\theta_{\varepsilon}(\tau-z) \tilde{K}_{t_0}(z;h) L(dz), \label{last}
\end{eqnarray}
{uniformly for $\tau \in \mathbb R$}. Notice that to pass from the first equation to the second we used (\ref{til}), and  the last identity follows from the Cauchy formula for analytic functions.

Now, with the same argument as for ${\mathcal I}_-$, we see that ${\mathcal I}_+=\mathcal{O}(\varepsilon^2 h^{\frac{M}{2}-n-2})$ uniformly for $\tau\in \mathbb R$ and $\varepsilon\in [h^{1-\delta}, \kappa[$ for all $\kappa>0$, which gives the result since $M>0$ is arbitrary. This ends the proof of Theorem \ref{th1}.

\vspace*{-0.7cm}
\begin{flushright}
$\square$
\end{flushright}

%{{\bf Remark 3.1.} Clearly the above proof works for $\theta \in C_0^{\infty}(]-1,-\frac{1}{2}[;\mathbb R)$ and then Theorem \ref{th1} remains valid in this case.}

{\bf Remark 3.1.} Let $O$ be an open bounded subset of $\mathbb C$ such that ${\rm supp}\tilde f\subset O$, and
assume that the function $K(z;h)$ defined by \eqref{formulaM} in the upper half plane
extends as a holomorphic function $\tilde K(z;h)$ to the zone $O_\ell(h):=O\cap\{\Im z\ge - {\ell}\zeta(h)\}$ for all $\ell\in \mathbb N$} and that the estimate
 $\tilde K(z;h) ={\mathcal O}(h^{-d(n)})$ holds uniformly for $z\in O_{{\ell}}(h)$ with $d(n)$ depending only on the dimension.
Then \eqref{main estimate 1} remains true uniformly  for $\varepsilon\in  [\kappa, h^{-\nu}[$ with fixed $\kappa>0$,  $\nu\in {\mathbb N}$.

To see this, we first see \eqref{i-}, since $\nu$ is fixed and $M$ is arbitrary.
Next, since ${\rm supp}\: \psi_L\subset\{z\in \mathbb C; \vert \Im z\vert   \leq  2 \frac{M}{\kappa} \zeta(h)   ) \} $  for all $\varepsilon \in [\kappa,h^{-\nu}[$, it follows  from the above assumption
(with $\ell > 2\frac{M}{\kappa}$) and the Cauchy formula that
\begin{eqnarray*}
{\mathcal I}_+(\tau,\varepsilon;h) &=& - \frac{1}{\pi} \int_{\{\Im z>0\}} \bar{\partial}\left( \tilde{f} \psi_L\right) (z) \mathcal{F}_{h}^{-1}\theta_{\varepsilon}(\tau-z)  K(z;h) L(dz) \\
&=& \frac{1}{\pi} \int_{\{\Im z<0\}} \bar{\partial}\left( \tilde{f} \psi_L\right) (z) \mathcal{F}_{h}^{-1}\theta_{\varepsilon}(\tau-z) \tilde K(z;h)L(dz).
\end{eqnarray*}
Then the same argument as for ${\mathcal I}_-$ shows ${\mathcal I}_+={\mathcal O}(h^\infty)$, and hence  \eqref{main estimate 1}  holds {uniformly for $\tau \in \mathbb R$ and $\varepsilon\in [\kappa, h^{-\nu}[$.}

Later, in the application to the study of the SSF, we shall show that assumption \eqref{deff} about the existence of an escape function implies that 
the function $z\mapsto K(z;h)$ (defined by \eqref{For2}) satisfies the condition assumed on $K(z;h)$ in this remark in an open complex  neighborhood  $O$ of $\tau_0$ (see Lemma \ref{ext an}). This will be crucial for the proof of the pointwise asymptotics \eqref{main semiclassical asymptotics}.

\subsection{Proof of Theorem \ref{th2}} Let $\varepsilon>0$ be a small constant  (independent of $h$) which will be fixed later. We have
$$
{\rm tr} \left( \chi^w \left[ f(H_j^w) \mathcal{F}_{h}^{-1} \theta_\varepsilon (\tau - H_j^w) \right]_0^1   \right)=I(\tau,\varepsilon;h)
$$
where ${\mathcal I}(\tau,\varepsilon;h)$ is defined by \eqref{basic expression1} with
$$
{K}(z;h)=  {\rm tr} \left( \chi^w [(z-H_j^w)^{-1}]_0^1\right ).
$$
It follows from  \eqref{til}   and   \eqref{-11} that, {uniformly for $\tau \in \mathbb R$},
\begin{equation}\label{th2.3}
{\mathcal I}(\tau,\varepsilon;h) \equiv  - \frac{1}{\pi} \int_{ \{L< |\Im z| < 2L \big\}} \bar{\partial}   \left(\tilde f\psi_L\right)(z) \mathcal{F}_{h}^{-1}\theta_{\varepsilon}(\tau-z) {\rm tr} \left( \chi^w [(z-H_j^w)^{-1}]_0^1  \right) L(dz).
\end{equation}
Let $B_0 \in C_0^\infty ({\mathbb R}^{2n})$ be a
real-valued function in the phase space  such that
$B_0=1 \hbox{ near }{\rm
supp\,}\chi$ and $B_0= 0\hbox{ near } {{\rm supp \,} }
{(H_1-H_0)}$, and
let $B=\alpha B_0$ for a constant $\alpha>0$ that we will choose later. We notice that the symbol $b=e^{B\log {1\over h}}$ is of class {$S_\delta ^l (\mathbb R^{2n})$}\footnote{Following \cite{dimassi}, $S^k_{\delta}(\mathbb R^{2n}):= \{a(\cdot;h)\in C^{\infty}(\mathbb R^{2n};\mathbb R);\: \forall \alpha\in \mathbb N^{2n}\, :\partial_{x,\xi}^{\alpha}a(x,\xi;h) = \mathcal{O}_{\alpha}(h^{-\delta|\alpha|-k})\}$, for $k\in \mathbb R$ and $\delta\in [0,1]$.} for $l= {\alpha \Vert B_0\Vert_{L^{\infty}(\mathbb R^{2n})}}$ and every $\delta >0$.  By
the same notation we also denote the
corresponding $h$-pseudodifferential operator, which is bounded, elliptic and has an inverse operator $(e^{B\log
{1\over h}})^{-1}$ with symbol in the same class. Using the $h$-pseudodifferential calculus (see \cite[Ch. 7]{dimassi}) as well as the Calder\'on-Vaillancourt theorem, it is clear that for some $k\in \mathbb N$,
\begin{equation}\label{exp}
e^{B\log {1\over h}}(z-H_1^w)(e^{B\log
{1\over h}})^{-1}=z-H_1^w+{\mathcal O}(\alpha \zeta(h))\Vert \nabla B_0\Vert _{C^k},
\end{equation}in
operator norm, for $h\le h(\alpha)$, where
$h(\alpha)>0$ is some continuous function. 

It follows that for  $\vert \Im z\vert\geq  {C\alpha }\zeta(h)$ (where $C$ depends only on $\Vert H_1\Vert_{C^k}$ and $\Vert B_0\Vert_{C^k}$), the right hand side of \eqref{exp} is invertible  and we have
%$\vert \alpha h\log
%{1\over h})\Vert \nabla G_0\Vert _{C^k} \vert \eq \vert \Im z\vert$ the left had side of \eqref{exp} is invertible and we ha
\begin{equation}\label{exp2}
\Vert e^{B\log {1\over h}}(z-H_1^w)^{-1}(e^{B\log
{1\over h}})^{-1}\Vert={\mathcal O}(\vert \Im z\vert^{-1}).
\end{equation}
in operator norm.

On the other hand, since $B=\alpha$ near ${\rm
supp\,}\chi$ and that $B=0$ near
${\rm supp}(H_1-H_0)$, it follows from the $h$-pseudodifferential calculus again that we have
$$
(e^{B\log {1\over h}})^{-1}(H_1^w-H_0^w)=(H_1^w-H_0^w)+{\mathcal O}  (h^\infty),
$$
$$
\chi^w e^{B\log {1\over h}}=e^{\alpha\log {1\over h}}    \chi^w+{\mathcal O}  (h^\infty),
$$
in {operator} norm and trace norm respectively. Thus
$$
e^{\alpha\log {1\over h}}   \,{\rm tr} \left( \chi^w(z-H_1^w)^{-1}(H_1^w-H_0^w)(z-H^w_0)^{-1}\right)
$$
$$={\rm tr}\left( \chi ^we^{B\log {1\over
h}}(z-H_1^w)^{-1}(e^{B\log {1\over h}})^{-1}(H_1^w-H_0^w)(z-H_0^w)^{-1}\right) +{\mathcal O}\left({h^\infty
\over
\vert \Im z\vert^2 }\right)  $$
$$={\rm tr}\left(\chi^w(z-e^{B\log {1\over h}} H_1^w(e^{B\log {1\over
h}})^{-1})^{-1} (H_1^w-H_0^w)(z-H^w_0)^{-1} \right)+{\mathcal O}\left({h^\infty \over
\vert \Im z\vert ^2} \right). 
$$
Combining this with \eqref{exp}, we deduce that for $\vert \Im z\vert\geq  C\alpha \zeta(h)$
\begin{equation}\label{th2.31}
{\rm tr}\left( \chi^w [(z-H_j^w)^{-1}]_0^1  \right) =  {\mathcal{O}\left( \frac{h^{\alpha-n}}{|\Im z|^2} \right) }.
\end{equation}

We choose $\alpha=\frac{M}{C\varepsilon}$. It follows from  \eqref{nouvelle formule}, \eqref{th2.3} and \eqref{th2.31}

{$${\mathcal I}(\tau,\varepsilon;h)={\mathcal O}\left(\varepsilon^3 h^{M(\frac{1}{\varepsilon C}-{2})-n-{3}} \log(\frac{1}{h})^{-2} \right).$$}

Next we choose $\varepsilon$ small enough so that $\frac{1}{\varepsilon C}>2$.  This ends the proof of Theorem \ref{th2} since $M$ is arbitrary. We recall that $C$ depends only on $\Vert H_1\Vert_{C^k}$ and $\Vert B_0\Vert_{C^k}$.

\vspace*{-0.7cm}
\begin{flushright}
$\square$
\end{flushright}

\subsection{Proof of Theorem \ref{WASG}}

Without any loss of generality,  we may assume that  $\chi$ is supported in a small neighborhood
of a fixed point $(x_0,\xi_0)$.  In fact  we may replace $\chi$ by a finite sum of terms $\chi\chi_i$ with $\sum_i \chi_i=1$ near the support of $\chi$ and $\chi_i$ has its support in a small neighborhood of a fixed point $(x_i,\xi_i)\in {\rm supp}\:\chi$. 
Then, choosing the support of {$\chi$} small enough,  we may assume that $\tau-H(x,\xi)$ is uniformly microhyperbolic in a fixed direction $T$ for $(x,\xi)$ in ${\rm supp}\:\chi$ and for $\tau$ near $\tau_0$.
Moreover, by modifying $H$ outside ${\rm supp}\:\chi$ as in Theorem \ref{main result appendix}, we may assume that 
$\tau-H(x,\xi)$ is uniformly microhyperbolic  in the whole phase space $\mathbb R^{2n}$
in the direction $T$ thanks to Theorem \ref{th2}.

Let $\theta\in C_0^{\infty}(]-1;1[;\mathbb R)$ be equal to one near $0$, $\varepsilon>0$ small enough independent of $h$ and $D$ be an integer such that $2^{-D}\sim h^{1-\delta}$ with $\delta\in ]0,\frac{1}{2}[$. Put $\tilde \varepsilon=2^{-D} \varepsilon$.
We write
$$\theta_{\varepsilon}({t})-\theta_{\tilde\varepsilon}({t})=\sum_{i=1}^D \Psi(2^{i-1} {t}),$$
where $\Psi({t})=\theta_{\varepsilon}({t})-\theta_{\varepsilon}(2 {t})\in C^\infty_0(\mathbb R)$.  Clearly, $\Psi({t})=\Psi_1({t})+\Psi_2({t})$ {where $\Psi_1$ and $\Psi_2$ are equal to $0$ near zero, ${\rm supp \:} \Psi_1\subset ]0,\varepsilon[$ and ${\rm supp \:}\Psi_2\subset ]-\varepsilon,0[$}. Now applying Theorem \ref{th1} (resp. Remark  2.1) to  $\Psi_1(2^{i-1} {t})$ (resp. $\Psi_2(2^{i-1} {t})$), $i=1,\cdots D,$ {we see that there exists $I\in O_{\tau_0}$ such that for all $f\in C_0^{\infty}(I;\mathbb R)$, we have} 
%Since  $\psi(x)=0$ near zero and ${\rm supp} \psi
%By Theorem  \ref{th1} we have 
\begin{equation}\label{D1}
{\rm tr}\; \left(  \chi^w  f (H^w) \mathcal{F}_{h}^{-1}\theta_{\varepsilon} (\tau-H^w)  \right) \equiv {\rm tr}\; \left(  \chi^w  f (H^w) \mathcal{F}_{h}^{-1}\theta_{\tilde\varepsilon} (\tau-H^w)  \right),
\end{equation}
{uniformly for $\tau \in \mathbb R$.}
%In fact we can represent the difference $(\theta-\theta_{\tilde\varepsilon})$ as a finite sum of function $\theta_\varepsilon$ appearing in Theorem \ref{th1}
%(with  $\varepsilon\in [\tilde\varepsilon, \frac{1}{C}[$) which with Remark leads  to  \eqref{D1}.
As in \eqref{D1111}, we have 
$$
{\rm tr}\; \left(  \chi^w  f (H^w) \mathcal{F}_{h}^{-1}\theta_{\tilde\varepsilon} (\tau-H^w)  \right) \equiv
$$
\begin{equation}\label{D2}
- \frac{1}{\pi} \int_{\{\vert \Im z\vert \geq h^\delta\}}    \bar{\partial}  \left( \tilde f \psi_{h^\delta} \right)(z) \mathcal{F}_{h}^{-1}\theta_{\tilde\varepsilon}(\tau-z) {\rm tr}\left( \chi^w (z-H^w)^{-1} \right) L(dz).
\end{equation}
Now in the zone $\Omega_{\delta}:= \{z\in {\rm supp}\: \tilde f; \vert \Im z\vert\geq h^\delta\}$, with $0< \delta< \frac{1}{2}$, the resolvent $(z-H^w)^{-1}$ is an  $h$-pseudodifferential operator. More precisely,  according to Proposition  8.6 in  \cite{dimassi}, there exists a $C^\infty$ matrix-valued function
$(x,\xi)\mapsto {\mathcal G}(x,\xi,z;h)$ such that
$$\Vert \partial_{x,\xi}^{\alpha} {\mathcal G}(x,\xi,z;h)\Vert \leq C_\alpha h^{-\delta(1+\vert \alpha\vert)}, \quad \forall \alpha \in \mathbb N^{2n},$$
uniformly on $z\in  \Omega_\delta$ and
\begin{equation}\label{Reso}
(z-H^w)^{-1}={\rm Op}_h^w({\mathcal G}(x,\xi,z;h)),
\end{equation}
for all $z\in \Omega_\delta$. Moreover
\begin{equation}\label{D4}
{\mathcal G}(x,\xi,z;h)\sim {\mathcal G}_0(x,\xi,z)+h{\mathcal G}_1(x,\xi,z)+h^2{\mathcal G}_2(x,\xi,z)+\cdots \quad {\text{in} \;\; S_{\delta}^{\delta}(\mathbb R^{2n},\mathcal{H}_N)}
\end{equation}
where ${\mathcal G}_j(x,\xi,z)$ is a finite sum of terms of the form
$$(z-H(x,\xi))^{-1} B_1(x,\xi,z)(z-H(x,\xi))^{-1}B_2(x,\xi,z)(z-H(x,\xi))^{-1}\cdots    B_k(x,\xi,z)(z-H(x,\xi))^{-1} ,$$
with $k<{2j+1}$, $B_l(x,\xi,z)\in S^0(\mathbb R^{2n};\mathcal{H}_N)$ holomorphic in $z$ near ${\rm supp\:}\tilde{f}$. Now by classical results on trace class $h$-pseudodifferential operators (see Theorem II.53 and Proposition II. 56 in \cite{Rob1}), we have for all {$m\in \mathbb N$},
\begin{equation}\label{D5}
{\rm tr}\; \left(  \chi^w  f (H^w) \mathcal{F}_{h}^{-1}\theta_{\varepsilon} (\tau-H^w)  \right) =
(2\pi h)^{-n}\sum_{j=0}^{{m}} a_j(\tau;h) h^j+{\mathcal O}(h^{(m+1)(1-2\delta)-n}),
\end{equation}
where
\begin{equation}\label{D5}
a_j(\tau;h)=- \frac{1}{\pi} \int \bar{\partial}  \left(    \tilde f       \psi_{h^\delta} \right)(z) 
\mathcal{F}_{h}^{-1}\theta_{\tilde \varepsilon}(\tau-z) \, \widehat e_j(z)\,  L(dz)
\end{equation}
with
$$\widehat e_j(z) := \iint_{{\mathbb R}^{2n}} \chi(x,\xi)\, \widehat{ {\rm tr}}   \left({\mathcal G}_j(x,\xi,z)\right)  dxd\xi .$$
Here $\widehat{\rm tr}$ denotes the  trace  of  square matrices.
The microhyperbolicity assumption implies that there exists $I\in O_{\tau_0}$ such that the function
 $$I\ni \tau\mapsto \widehat e_j(\tau\pm i0):=
\lim_{s\searrow 0} \widehat e_j(\tau\pm is),$$
is $C^\infty$  (see Proposition \ref{Hyperb}). 
Set
\begin{equation}\label{D7}
\gamma_j(\tau):=-\frac{1}{2\pi i}   \Big[ \widehat{e}_j(\tau+i0)- \widehat{e}_j(\tau-i0)\Big].
\end{equation}
Now the following lemma ends the proof of Theorem \ref{WASG}.

\begin{lemma}
$$a_j(\tau;h)= f(\tau)  \gamma_j(\tau)+{\mathcal O}(h^\infty).$$
\end{lemma}

\begin{proof}
Since $z\mapsto 
\mathcal{F}_{h}^{-1}\theta_{\tilde{\varepsilon}}(\tau-z) \, \widehat e_j(z)$ is  holomorphic in the complex domain $\{z\in \mathbb C; \pm \Im z>0\}$, it follows from the Green formula
that
$$a_j(\tau;h)=- \frac{1}{\pi}\lim_{s \searrow 0}  \int _{\{\Im z>s\}}\bar{\partial}  \left(    \tilde f\psi_{h^\delta} \right)(z) 
\mathcal{F}_{h}^{-1}\theta_{\tilde \varepsilon}(\tau-z) \, \widehat e_j(z)\,  L(dz)$$
$$- \frac{1}{\pi}\lim_{s \searrow 0}  \int _{\{\Im z< - s\}}\bar{\partial}  \left(    \tilde f       \psi_{h^\delta} \right)(z) 
\mathcal{F}_{h}^{-1}\theta_{\tilde\varepsilon}(\tau-z) \, \widehat e_j(z)\,  L(dz)$$
$$=-\frac{1}{2\pi i}\int_{\mathbb R} f(\lambda)\mathcal{F}_{h}^{-1}\theta_{\tilde \varepsilon}(\tau-\lambda)
\Big[\widehat{e}_j(\tau+i0)- \widehat{e}_j(\tau-i0) \Big] d\lambda$$
$$=\int_{\mathbb R} \mathcal{F}_{h}^{-1}\theta_{\tilde \varepsilon}(\tau-\lambda) f(\lambda)
\gamma_j(\lambda)d\lambda=\int_{\mathbb R} \mathcal{F}_1^{-1}\theta(\lambda) f(\tau - h^\delta \lambda) \gamma_j(\tau - h^\delta \lambda)d\lambda.$$
The last equality is obtained by a change of variable. Applying Taylor's formula to the function $\lambda\mapsto f(\tau - h^\delta \lambda)\gamma_j(\tau - h^\delta \lambda)$ at $\lambda=0$ and using the fact that
$\int  \mathcal{F}_1^{-1}\theta(\lambda) (-i\lambda)^kd\lambda=\theta^{(k)}(0)=0$ we obtain the lemma. We recall that $\theta=1$ near zero. Here ${\mathcal F}_1^{-1}$ is ${\mathcal F_h}^{-1}$ with $h=1$
\end{proof}

\section{Proofs of the results on the SSF}\label{Proofs of the results SSF}
{This section is devoted to the proofs of the results of subsection \ref{SSF}.

We follow the notations used in section 3. From the assumption \eqref{DA1}, the operator  
$$\big[(P_{j}(h)-z_0)^{-q}(z-P_{j}(h))^{-1} \big]_0^1$$
is of trace class for $q> \frac{n}{2}$ and $z_0\not\in \sigma (P_1(h))\cup\sigma (P_0(h))$ which are fixed in what follows. 

We set
\begin{equation}\label{For2}
K(z;h):=(z-z_0)^q {\rm tr}\bigg( \big[(P_{j}(h)-z_0)^{-q}(z-P_{j}(h))^{-1} \big]_0^1 \bigg), \quad \Im z\not =0.
\end{equation}

As in the proof of   \eqref{HS},  formulas      \eqref{LK}  and \eqref{HeSj} yield, for all $f,\theta\in C_0^{\infty}(\mathbb R;\mathbb R)$,
\begin{equation}\label{W1}
\langle s_{h}'(\cdot),f(\cdot) \rangle =
  \frac{1}{\pi} \int_{\mathbb C} \bar{\partial} \tilde{f}(z) K(z;h) L(dz),
\end{equation}
\begin{equation} \label{weak dim}
\langle s_{h}'(\cdot), \mathcal{F}_{h}^{-1}\theta(\tau- \cdot)f(\cdot)\rangle =
  \frac{1}{\pi} \int_{\mathbb C} \bar{\partial} \tilde{f}(z) \mathcal{F}_{h}^{-1}\theta(\tau-z) K(z;h) L(dz).
\end{equation}

\subsection{Proof of Theorem \ref{WA1}} 
This is a classical result and follows from the functional calculus of $h$-pseudodifferential operators.

Let $\delta\in ]0, \frac{1}{2}[$. 
The contribution from the domain $\vert {\rm
Im\,}z\vert\le h^\delta$ of the integral in the right
hand side of \eqref{W1} is ${\mathcal O}(h^\infty)$.

 Next, in the domain
$\vert {\rm Im\,}z\vert\ge h^\delta$,} we use the fact that $(z-P_k(h))^{-1}$ are $h$-pseudodifferential operators, {$k=0,1$} 
(see \eqref{Reso} and \eqref{D4}). This formally yields \eqref{Oubl} (with $q=0$) with
$$c_j(f)=\iint_{{\mathbb R}^{2n}} \, \frac{1}{\pi} \widehat{ {\rm tr}}  \bigg( \int_{\mathbb C} \bar{\partial} \tilde{f}(z)   \left({\mathcal G}_{j,1}(x,\xi,z)-{\mathcal G}_{j,0}(x,\xi,z)\right) L(dz) \bigg)dxd\xi.$$
In particular
$$c_0(f)= \iint_{\mathbb R^{2n}} \widehat{{\rm tr}}\big(f({p_{0}(x,\xi)})-f({p_{1}(x,\xi)})\big) dxd\xi,$$
and \eqref{Oubl2}  trivially follows from this
formula.  
To see that $c_{2j+1}(f)=0$, it suffices to notice that 
$h\mapsto \vert 2\pi h\vert^n{\rm tr}(f(P_1(h))-f(P_0(h))$ is an even function. More rigorously for $q\neq 0$, one may write $K(z;h)$ as
\begin{align}
K(z;h) & = (z-z_0)^q {\rm tr} \: \left[ \big( (P_1(h)-z_0)^{-q} - (P_0(h)-z_0)^{-q} \big) (z-P_1(h))^{-1} \right] \nonumber \\
&+ (z-z_0)^q  {\rm tr} \: \left[ (P_0(h)-z_0)^{-q} \big( (z-P_1(h))^{-1} - (z-P_0(h))^{-1} \big) \right] \label{nouvelle forme},
\end{align}
and use the fact that $(P_k(h)-z_0)^{-q}$ are $h$-pseudodifferential operators, $k=0,1$. This ends the proof of  \eqref{Oubl}.

 \subsection{Proof of Theorem \ref{ASW2}} 
The proof of Theorem \ref{ASW2}  uses \eqref{weak dim} and 
is quite similar to that of Theorem \ref{WASG}, and we omit the details. The main difference is that  $\Sigma_{\tau_0}=\{(x,\xi)\in {\mathbb R}^{2n}; \, {\rm det} (p_1(x,\xi)-\tau_0)=0\}$ is not a compact set in ${\mathbb R}^{2n}$. In this case we have to justify that we can cover $\Sigma_{\tau_0}$ by finite
open sets $O_1, O_2,\cdots,O_\ell$  in which we can construct $\tilde p_{1,k}(x,\xi)$, and $T_k\in {\mathbb R}^{2n}$, $k=1,\cdots, \ell$, 
such that $\tilde p_{1,k}(x,\xi)-\tau_0$
 is uniformly microhyperbolic in the  direction $T_k$ and $\tilde p_{1,k}(x,\xi)=p_1(x,\xi)$ for all $(x,\xi)\in O_k$. To see this, we first notice that
 $\Sigma_{\tau_0}=\Sigma_{\tau_0}\cap\{\vert \xi\vert \leq R_0\}$ ($R_0$ being large enough), since $\lim_{\vert \xi \vert \rightarrow \infty} {\rm det}(p_1(x,\xi)-\tau_0)=\infty$. Next, fix $R_1$ large such that  ${\rm inf}_{\vert x\vert>R_1}\vert {\rm det} (V(x)-\tau_0)\vert>0$.
This is possible since {$\lim_{|x|\rightarrow \infty} V(x) = V_{\infty}$ and} $\tau_0\not\in\sigma(V_\infty)$ by assumption. On the compact set $\Sigma_{\tau_0}\cap \{\vert x\vert \leq R_1\}$  we can apply Theorem \ref{WASG} without any modification. On the other hand,  we see from the choice  of $R_1$  that$\nabla(\vert \xi\vert^2)\not =0$ for all $(x,\xi)\in \Sigma_{\tau_0}\cap \{\vert x\vert >R_1\} := \Sigma_{\tau_0,R_1}$.
Thus,    we can find finite open covers $o_1,o_2,\cdots o_\ell$ in ${\mathbb R}^n$, $\tilde T_1,\tilde T_2,\cdots, \tilde T_\ell\in {\mathbb R}^n$ and $c_1,\cdots c_\ell>0$ 
such that $ \{\vert \xi\vert \leq R_0\}\subset \bigcup_{j=1}^\ell o_j$ and 
  for each $j=1, \cdots \ell$, 
%construct a finite small bounded  open set ,  such that
%and
  $\langle \tilde T_j,\nabla(\vert \xi\vert^2)\rangle \geq c_j$,  uniformly on $\xi\in o_j\cap \pi_\xi \Sigma_{\tau_0,R_1}$.
 Now using Theorem \ref{main result appendix},  we construct $\tilde p_{1,j}(x,\xi)$, 
such that $\tilde p_{1,j}(x,\xi)-\tau_0$
 is uniformly microhyperbolic   in the  direction ${T_j=(0,\tilde T_j)}$ and $\tilde p_{1,j}(x,\xi)=p_1(x,\xi)$ for all $(x,\xi)\in \{\vert x\vert >R_1\}\times o_j$.
 We can now proceed analogously to the proof  of Theorem \ref{WASG}.

\subsection{Proof of Theorem \ref{WRA}} 
For the proof of Theorem \ref{WRA}, {assume that $\tau\mapsto s(\tau;h)$ is monotonic (i.e., $s'(\cdot;h)$ is positive or negative in the  sense of  distributions).}  In this case Theorem \ref{WRA} is a simple consequence of Theorem \ref{ASW2} by standard Tauberian arguments (see  \cite{dimassi}, 
\cite{iv}, \cite{Rob1}).
 
 For the general case, we use a trick due to Robert \cite{Robert0}, which consists in writing $s(\tau;h)=s_1(\tau;h)-s_2(\tau;h)$
where $\tau\mapsto s_i(\tau;h)$, $i=1,2$ are monotonic. Now,  it suffices to apply   the above argument to each $s_i(\tau;h)$.

Notice that, Robert's trick applies to Schr\"odinger operators with matrix-valued potential under the assumption \eqref{DA1}  with scalar matrix $V_\infty$.

\subsection{Proof of Theorem \ref{main semiclassical}} 

%In the following we  assume that $V(x)-V_\infty\in C^\infty_0({\mathbb R}^n;M_N({\mathbb C}))$.  We refer to Remark 4.1  for the general case.
The proof of the following lemma is the same as  that of  Lemma 2.2   in  \cite{Dimassi2}.
 %Theorem is a simple consequence of Theorem  
\begin{lemma}\label{RepL}
Under  the assumption \eqref{DA1}, we have 
$$s_h'(\tau)= \frac{1}{\pi} \Im
K(\tau+i0;h) \text { in } {\mathcal D}'({\mathbb R}),
$$
i.e.  we have, for all $f\in { C}^\infty_0(\mathbb R)$, 
$$
\displaystyle \langle s_h'(\cdot),f\rangle =\lim_{\kappa
\searrow 0}\frac{1}{\pi}\int_{\mathbb R} f(\tau)\Im
K(\tau+i\kappa;h)\,d\tau.
$$ 
\end{lemma}
 
Let $I\in O_{\tau_0}$ such that \eqref{deff} holds on $\Sigma_{\tau}$ for all $\tau\in I$. For $M\geq 0$, we introduce the following $h$-dependent set
\begin{equation}\label{ensembles}
\Gamma_M : = \left\{z\in \mathbb C; \Re z \in I\; \text{and}\; \Im z> - M
\zeta(h)\right\},
\end{equation}
where we recall that $\zeta(h)=h\log(\frac{1}{h})$. 

The idea of the proof of the following lemma is based on the theory of resonance and close to the one of Theorem 1  in \cite{Sj}. %the main  step in the proof of Theorem \ref{main semiclassical}.
\begin{lemma}\label{ext an}
In addition to the assumptions \eqref{DA1} and \eqref{deff}, we assume that $V(x)-V_\infty\in C^\infty_0({\mathbb R}^n; {\mathcal H}_N)$. For any $M>0$, the function $z\mapsto K(z;h)$ has an analytic extension from $\Gamma_0$ to $\Gamma_M$. Moreover, we have, uniformly for $z\in \Gamma_M$,
\begin{equation}\label{sk1}
K^{(k)}(z;h) =\mathcal{O} \big(
h^{-n}\zeta(h)^{-k-1}\big), \quad \forall k\in \mathbb N.
\end{equation}
In particular, uniformly for $\tau\in I$,
\begin{equation}\label{equuu}
s_{h}^{(k+1)}(\tau) = \mathcal{O}\big(h^{-n}\zeta(h)^{-k-1}\big), \quad \forall k\in \mathbb N.
\end{equation}
\end{lemma}
\begin{proof}
The estimate (\ref{equuu}) follows immediately from (\ref{sk1}) and the representation of the SSF given by Lemma \ref{RepL}.
Hence it is enough to prove \eqref{sk1}.

Let $F : \mathbb R^n \rightarrow \mathbb R^n$ be a smooth vector field such that $F = 0$ in a neighbourhood of $\text{supp}(V-V_{\infty})$ and $F(x)=x$ for $|x|$ large enough. For $\omega\in\mathbb R$  small enough, we denote $U_\omega : L^2({\mathbb R}^n; \mathbb C^N)\rightarrow L^2({\mathbb R}^n, \mathbb C^N)$ the unitary operator defined by 
\begin{equation}\label{unitary distortion}
U_{\omega} \phi(x):= | \text{det}(1+\omega \nabla F(x))   |^{\frac{1}{2}} \phi(x+\omega F(x)),
\end{equation}
and set
$$
P_{j,\omega}(h):=U_{\omega}P_j(h) (U_{\omega})^{-1},\quad j=0,1.
$$
They are differential operators with analytic coefficients with respect to $\omega$, and can
 be analytically continued to small enough complex values of $\omega$.
 It follows from the analytic perturbation theory (see \cite{Kato})  that  for $\omega_0$ small enough,
  $\omega \in ]-\omega_0,\omega_0[\:\mapsto P_{j,\omega}(h)$, $j=0,1$,  extends to an analytic  type
  ${\mathcal A}$-family  of operators on %
  $D(\omega_0):=\{\omega\in \mathbb C; \vert \omega\vert <\omega_0\}$ with domain $H^2({\mathbb R}^n;\mathbb C^N)$.

We set, first for real $\omega$ and   $\Im z>0$,
\begin{equation}\label{anal}
K_{\omega}(z;h):=(z-z_0)^q {\rm tr }\bigg( \big[(P_{j,\omega}(h)-z_0)^{-q}(z-P_{j,\omega}(h))^{-1} \big]_0^1 \bigg).
\end{equation}
%Since the operator $(P_{j}(h)-z_0)^{-q}(z-P_{j}(h))^{-1}$ is unitary equivalent to $(P_{j,\theta}(h)-z_0)^{-q}(z-P_{j,\theta}(h))^{-1}$, it follows from 
Since $U_\omega$ is unitary for real  $\omega$, 
it follows from the cyclicity of the trace that
\begin{equation}\label{egalite pour teta reelle}
K(z;h)=K_{\omega}(z;h), \quad \forall \omega \in ]-\omega_0,\omega_0[, \forall\;  \Im z>0.
\end{equation}

On the other hand, for $\Im z>M_0\zeta(h)$ for a given $h$-independent $M_0>0$, the function $\omega  \mapsto K_{\omega}(z)$
 is analytic in $\omega\in D(2c M_0\zeta(h))$ with some $c>0$ independent of $h$ and $M_0$.
 Thus,  by the uniqueness theorem of analytic continuation, the identity \eqref{egalite pour teta reelle} remains true for  $\Im z>M_0\zeta(h)$ and $\omega\in D(2c M_0\zeta(h))$, i.e.,
\begin{equation}\label{l}
K(z;h)=K_{\omega}(z;h), \quad \forall \omega \in D\big(2c M_0\zeta(h)), \forall \; \Im z> M_0\zeta(h).
\end{equation}
From now on we fix $M = c M_0$ and set $\omega_1=iM\zeta(h)$. 
%In particular, from the above equation, we have
%\begin{equation}
%\sigma_{h}(z)=\sigma_{h,\theta_1}(z), \quad  \forall \; \Im z> \tilde{M}h \log(1/h).
%\end{equation}

Since $\tau_0  > {e_{N,\infty}}$, $x\cdot \xi$ is an escape function for $p_1(x,\xi)$ for 
$\vert(x,\xi)\vert$ 
large enough. Thus, without any loss of generality, we may assume that $G(x,\xi)= x\cdot \xi$ for $\vert(x,\xi)\vert$ large enough. Then $\tilde G(x,\xi):=G(x,\xi)-F(x)\cdot \xi$
has a compact support, and in particular its quantization $\tilde{G}^w(x,h D_x)$ is $L^2$-bounded by the Calder\'on-Vaillancourt theorem and the operators $e^{\pm\frac{M\zeta(h)}{h} \tilde{G}^w(x,h D_x)}$ are
well-defined.

Let us define
$$
\tilde{P}_{j,\omega_1}(h):=e^{-\frac{M\zeta(h)}{h}\tilde{G}^w(x,h D_x)} P_{j,\omega_1}(h) e^{\frac{M\zeta(h)}{h} \tilde{G}^w(x,h D_x)}, \quad j=0,1,
$$
\begin{equation}\label{steta}
\tilde K_{\omega_1}(z;h):=(z-z_0)^q {\rm tr}\bigg(\bigg[\big(\tilde{P}_{j,\omega_1}(h)-z_0\big)^{-q}\big(z-\tilde{P}_{j,\omega_1}(h)\big)^{-1} \bigg]_0^1\bigg).
\end{equation}
From Lemma \ref{4.3} below, 
$z \mapsto \tilde{K}_{\omega_1}(z;h)$ is analytic in $\Gamma_{c M}$ for some $c>0$. Again by the cyclicity of the trace and the uniqueness of the analytic continuation, 
we conclude 
\begin{equation}\label{var}
\tilde{K}_{\omega_1}(z;h)=K(z;h), \quad \forall z\in \Gamma_{cM}.
\end{equation}
This with the resolvent estimate \eqref{main estimation2} leads to
$$K(z;h)={\mathcal O}(h^{-n} \zeta(h)^{-1}),$$
uniformly for $z\in \Gamma_{c M}$, which yields \eqref{sk1} for $k=0$.
Next, taking the derivative of \eqref{steta} and applying \eqref{main estimation2} we obtain 
\eqref{sk1} for $k \geq 1$ (Recall that the trace of semiclassical quantization of a symbol in a suitable class is of ${\mathcal O}(h^{-n})$, see \cite[Theorem 9.4]{dimassi}).
\end{proof} 

\begin{lemma}
\label{4.3}
There exists $c>0$ such that  for all $M>0$ the operator  $\tilde P_{j,\omega_1}(h)-z$ is invertible for every $z\in \Gamma_{c M}$. Moreover,
one has, uniformly in this domain,
\begin{equation}\label{main estimation2}
\Vert \big(z-\tilde{P}_{j,\omega_1}(h)\big)^{-1}\Vert = \mathcal{O}\big(\zeta(h)^{-1}\big).
\end{equation}
\end{lemma}
\begin{proof}
%Using the notation $\rm ad_A B := [A,B] = A\circ B - B\circ A$, we have
We have
 \begin{equation}\label{Oubl4}
\tilde{P}_{j,\omega_1}(h)=e^{{-}\frac{M\zeta(h)}{h}{\rm ad}_{\tilde G^w}} {P_{j,\omega_1}(h)} \sim\sum_{k=0}^\infty\frac{ {(-M\zeta(h))^k}}{k!} \bigg(\frac{1}{h}{\rm ad}_{\tilde G^w}\bigg)^kP_{j,\omega_1}(h),
\end{equation}
where ${\rm ad}_{\tilde G^w} P_{j,\omega_1}(h) =[\tilde G^w, P_{j, \omega_1}(h)]={\mathcal O}(h)$ \footnote{We have used the fact  $G$ is scalar valued only to prove that 
$[\tilde G^w, P_{j, \omega_1}(h)]={\mathcal O}(h)$.}.
By definition, $\zeta(h)$ tends to $0$ as $h\searrow 0$. Combining this with the boundedness of
$h^{-1}{\rm ad}_{\tilde G^w}$ we find 
 that the asymptotic expansion \eqref{Oubl4} makes sense. In particular,
\begin{equation*}\label{developpement}
\tilde{P}_{j,\omega_1}(h) = P_{j,\omega_1}(h) -\frac{M\zeta(h)}{h} \big[ \tilde{G}^w(x,h D_x),P_{j,\omega_1}(h) \big]+\mathcal{O}(M^2\zeta(h)^2).
\end{equation*}  
Let $p_{j,\omega_1}$, $\tilde p_{j,\omega_1}$ be the Weyl symbols of $P_{j,\omega_1}$, $\tilde P_{j,\omega_1}$ respectively. 
We obtain  from the $h$-pseudodifferential calculus (\cite[Ch. 7]{dimassi}),
\begin{equation}\label{operateur conjugue}
\tilde{p}_{j,\omega_1} = p_{j,\omega_1} - i M\zeta(h)\{p_{j,\omega_1}, \tilde{G}\}+\mathcal{O}
(M^2\zeta(h)^2),
\end{equation}
and in particular, using the Taylor expansion of $p_{j,\omega_1}$ with respect to $\omega_1$;
$$
p_{j,\omega_1}=p_j-iM\zeta(h)\{p_j,F(x)\cdot \xi\}+{\mathcal O}(M^2\zeta(h)^2),
$$
we obtain
\begin{equation}\label{Im}
\Im (\tilde{p}_{j,\omega_1})= -M\zeta(h)\{{p_j}, \tilde{G}+F(x)\cdot\xi\}+\mathcal{O}(M^2\zeta(h)^2).
\end{equation}
\begin{equation}\label{Rel}
{\Re (\tilde{p}_{j,\omega_1})=p_j + {\mathcal O}(M\zeta(h)).}
\end{equation}
Since  $G(x,\xi)=\tilde G(x,\xi)+ {F(x)}\cdot \xi$ satisfies the assumption \eqref{deff},  it follows from \eqref{Im} and \eqref{Rel} that there exist $C>0$ and $I\in O_{\tau_0}$ such that
\begin{equation}\label{main est}
-\Im (\tilde{p}_{1,\omega_1})(x,\xi) \geq C M\zeta(h), \quad \forall (x,\xi)\in \Sigma_{I}:=\bigcup_{\tau\in I}\Sigma_{\tau},
\end{equation}
Of course, the same estimate holds also for $\Im (\tilde{p}_{0,\omega_1})(x,\xi)$, since \eqref{deff} always holds for $p_0$ with $G=x\cdot \xi$ for any $\tau_0>e_{N,\infty}$.

We write $\tilde{P}_{j,\omega_1}(h)-z=A_{j,\omega_1}(h)-\Re z+i (B_{j,\omega_1}(h)-\Im z)$ with
$$
A_{j,\omega_1}(h) = \frac12\bigg( \tilde{P}_{j,\omega_1}(h) + \big( \tilde{P}_{j,\omega_1}(h)\big)^*\bigg), 
\quad 
B_{j,\omega_1}(h)=\frac1{2i} \bigg(\tilde{P}_{j,\omega_1}(h)-\big(\tilde{P}_{j,\omega_1}(h)\big)^*\bigg).
$$

Let $\psi_{j,1},\psi_{j,2} \in C^{\infty}(\mathbb R^{2n};\mathbb R)$ be such that, for $I'\Subset I$,
$$
\psi_{j,1}^2+\psi_{j,2}^2=1, \quad {\psi_{j,1}}= 1 \; \text{on}\; \Sigma_{I'}^{j}, \quad  \text{supp}(\psi_{j,1}) \subset  \Sigma_{I}^{j}.
$$
According to Lemma 3.2 in \cite{Sj}, there exist two self-adjoint operators $\Psi_{j,1}$ and $\Psi_{j,2}$ with principal symbols respectively $\psi_{j,1}$ and $\psi_{j,2}$ such that
\begin{equation}\label{Ident}
(\Psi_{j,1})^2 + (\Psi_{j,2})^2={\rm  Id} + \mathcal{O}(h^{\infty}) \quad \text{in}\;\; \mathcal{L}(L^2(\mathbb R^n)).
\end{equation}
We denote by the same letters the operators $\Psi_{j,i}:=\Psi_{j,i} I_N$, $i=1,2$. On the support of $\psi_{j,1}$,  we see from (\ref{main est}) that the principal symbol of $-B_{j,\omega_1}(h)$ is estimated from below by $ C M \zeta (h)$.
Then by the G\aa rding's inequality,
we obtain, uniformly for $\Im z>-\frac{C}{3}M\zeta(h)$,
\begin{equation}\label{Gard}
\Vert (\tilde{P}_{j,\omega_1}(h)-z) \Psi_{j,1} u \Vert \cdot \Vert \Psi_{j,1} u\Vert \geq |\langle (\tilde{P}_{j,\omega_1}(h) -z)\Psi_{j,1} u, \Psi_{j,1} u   \rangle | 
\end{equation}
$$
\geq  | \langle (\Im \tilde{P}_{j,\omega_1}(h)-\Im z) \Psi_{j,1} u, \Psi_{j,1} u\rangle| 
= \langle (\Im z-B_{j,\omega_1}(h)   ) \Psi_{j,1} u, \Psi_{j,1} u \rangle  $$
$$
\geq ( {\Im z+CM\zeta(h)-\mathcal O(h))\Vert \Psi_{j,1} u \Vert^2\geq  \frac{C}{3} M\zeta(h)\Vert \Psi_{j,1} u \Vert^2}.
$$

On the other hand,  since $A_{j,\omega_1}(h)-\Re z$ is uniformly elliptic on the support of $\psi_{j,2}$
and $\Re z\in I$, the  symbolic  calculus permits us to construct   a  parametrix  $R\in S^0(\langle \xi\rangle^{-2}) $ of  $A_{j,\omega_1}(h)-\Re z$ such that, in the symbol sense,
$$R\# (A_{j,\omega_1}(h)-\Re z)\psi_{j,2}=\psi_{j,2}+{\mathcal O}(h^\infty),$$
where  $\#$ stands for the Weyl composition of symbols.  As a consequence, we obtain
\begin{eqnarray}\label{ell}
\Vert( \tilde{P}_{j,\omega_1}(h)-z) \Psi_{j,2} u \Vert \geq \frac{1}{C'} \Vert \Psi_{j,2} u \Vert-{\mathcal O}(h^\infty)\Vert u\Vert^2.
\end{eqnarray}
Furthermore,  by means of standard  elliptic arguments, one can easily prove  the following semiclassical  inequality 
% result operators give for any $u\in H^2({\mathbb R}^n; \mathbb C^m)$  
\begin{equation}\label{oubli6}
\Vert [\tilde{P}_{j,\omega_1}(h), \Psi_{j,i}] u \Vert \leq C_2h (\Vert \tilde{P}_{j,\omega_1}(h) u \Vert+\Vert u\Vert), \quad \forall u\in H^2({\mathbb R}^n; \mathbb C^N).
\end{equation}
%{\cb je veux ajouter quelques phrases pour controler les commutateurs $[\Psi_j\tilde P_{j,\theta} ]$ et il y'aura une in\'egalit\'e suppl\'ementaire}
Combining \eqref{Ident},  \eqref{Gard},  (\ref{ell}), and \eqref{oubli6}   with the estimate
\begin{equation}\label{Oubli1234}
\Vert (\tilde{P}_{j,\omega_1}(h)-z)u \Vert^2 {=} \sum_{i=1}^2\Vert \Psi_{j,i}(\tilde{P}_{j,\omega_1}(h)-z) u \Vert^2    -  {\mathcal   O}( h^{\infty})      \Vert (\tilde{P}_{j,\omega_1}(h)-z) u \Vert^2 
\end{equation}
$$
\geq \frac{1}{2}\sum_{i=1}^2 \Vert (\tilde{P}_{j,\omega_1}(h)- z)\Psi_{j,i} u \Vert^2-\sum_{i=1}^2\Vert [\tilde{P}_{j,\omega_1}(h), \Psi_{j,i}] u \Vert^2 -{\mathcal  O}( h^{\infty})   \Vert (\tilde{P}_{j,\omega_1}(h)-z) u \Vert^2 ,
$$
 we deduce, for  $z\in \Gamma_{{ c M}}$ (with $c>0$ independent of $M$ and $h$) and sufficiently small $h$,
\begin{equation}\label{Est}
\Vert (\tilde{P}_{j,\omega_1}(h)-z)  u\Vert \geq \frac{\zeta(h)}{C} \Vert u\Vert.
\end{equation}
%uniformly on $z\in \Gamma_{M_1}$, with $CM_1=M$. 
 By the same arguments, we prove an estimate similar to \eqref{Est}  for the adjoint operator
$\tilde{P}_{j,\omega_1}(h)^*-\overline{z}$ and 
we conclude that $\tilde{P}_{j,\omega_1}(h)-z$ is invertible for every $z\in \Gamma_{cM}.$  Moreover \eqref{Est} yields
the resolvent estimate \eqref{main estimation2}.
\end{proof}

\textbf{End of the proof of Theorem \ref{main semiclassical}.} 

Using Lemma \ref{ext an}, and applying  Theorem \ref{th1} and Remark 3.1  to the right hand  side of  \eqref{weak dim} we obtain  the following lemma.
\begin{lemma}\label{Est23}
Assume that $\varphi \in C_0^{\infty}({]-1,1[};\mathbb R)$ is 0 in a neighborhood of 0.
Let $\kappa$ be  a positive constant independent of $h$ and $\nu\in \mathbb N$. {Under the assumptions of Lemma \ref{RepL}}, there exists $I\in O_{\tau_0}$ such that  for 
$f\in C_0^{\infty}(I;\mathbb R)$, we have 
\begin{equation}\label{main esttt}
\langle  s_{h}'(\cdot), \mathcal{F}_{h}^{-1}\varphi_{\varepsilon} (\tau - \cdot) f(\cdot)\rangle = \mathcal{O}(h^{\infty}),
\end{equation}
uniformly for $\tau\in {\mathbb R}$ and $\varepsilon\in ]\kappa,h^{-\nu}[$.
\end{lemma}

Now let $\theta\in C_0^{\infty}(]-1, 1[;\mathbb R)$  be equal to one near $0$ and let  $f\in C_0^{\infty}(I;\mathbb R)$ be as in the above lemma.  
Suppose  $\varepsilon>0$ is  a small enough constant independent of $h$ and $\tilde\varepsilon:= h^{-\nu}$ with $\nu \in \mathbb N$ arbitrary large.

{Repeating the same construction as in the proof of Theorem \ref{WASG}, we represent the difference $\theta_{\tilde\varepsilon}-\theta_{\varepsilon}$ as a finite sum $\sum_{0\leq j\leq N(h)} \varphi_{\varepsilon_j}$ with $\varphi_{\varepsilon_j}$ as in Lemma \ref{Est23} and
$N(h)= \mathcal{O}(h^{-\nu})$.   Applying Lemma \ref{Est23}   to each term,  we get 
\begin{equation}\label{ecr}
\langle s_{h}'(\cdot), \mathcal{F}_{h}^{-1}\theta_\varepsilon(\tau-\cdot)f(\cdot)\rangle = \langle s_{h}'(\cdot), \mathcal{F}_{h}^{-1}\theta_{\tilde\varepsilon}(\tau-\cdot)f(\cdot)\rangle + \mathcal{O}(h^{\infty}),
\end{equation}
uniformly with respect to $\tau\in \mathbb R$.

Next, by a change of variable we have
$$\langle s'_{h}(\cdot), \mathcal{F}_{h}^{-1}\theta_{\tilde\varepsilon}(\tau-\cdot)f(\cdot)\rangle =\int_{\mathbb R} {\mathcal F}_1^{-1}\theta(\lambda) (fs'_h)(\tau-h^{1+\nu} \lambda)d\lambda.$$
%where ${\mathcal F}_1^{-1}$ is ${\mathcal F_h}^{-1}$ with $h=1$.
Applying Taylor's formula to the function $\lambda \mapsto (fs'_h)(\tau-h^{1+\nu} \lambda)$ at $\lambda=0$ and using  \eqref{equuu} with {$k=1$},
we get
%On the other hand, a simple computation using Taylor's formula and the estimate (\ref{equuu}) yields
\begin{equation}\label{fin}
\langle s'_{h}(\cdot), \mathcal{F}_{h}^{-1}\theta_{\tilde\varepsilon}(\tau-\cdot)f(\cdot)\rangle = s_{h}'(\tau)f(\tau)+ \mathcal{O}(h^{\nu+1-n}\zeta(h)^{-2}),
\end{equation}
uniformly for $\tau\in \mathbb R$ since $ \int_{\mathbb R}   \mathcal{F}_1^{-1}\theta(\lambda)d\lambda=\theta(0)=1$. % \in I_{\lambda_0}(\eta)$.

From (\ref{ecr}) and (\ref{fin})   we deduce 
\begin{equation}\label{end}
s_{h}'(\tau)f(\tau) = \langle s_{h}'(\cdot), \mathcal{F}_{h}^{-1}\theta_\epsilon(\tau-\cdot)f(\cdot)\rangle + \mathcal{O}(h^{\nu+1-n}\zeta(h)^{-2}).
\end{equation}
By Theorem \ref{ASW2},   the first term of  the right hand side of the above equality has an asymptotic expansion in powers of $h$. Now, 
since $\nu$ is arbitrary, {the asymptotic expansion \eqref{main semiclassical asymptotics}} follows from   \eqref{end} {by choosing $f$ equal to $1$ near $\tau_0$}. This ends the proof  of Theorem \ref{main semiclassical}  under the assumption $V-V_\infty\in C^\infty_0({\mathbb R}^n; \mathcal{H}_N)$.

\vspace*{-0.7cm}

\begin{flushright}
$\square$
\end{flushright}

{\bf Remark 4.1.}
Notice that, except for Lemma  \ref{ext an}, all the steps of the proof of Theorem \ref{main semiclassical}  remain valid under the assumptions \eqref{DA1}  and  \eqref{deff} with $\mu>n$.
We will now show how to dispense with the assumption on the support of $V$ in Lemma \ref{ext an}. According to Proposition 4.2 in  \cite{MRS}, if $V$ satisfies {\eqref{DA1}}, then  for any $\varkappa>0$ and  $ \tilde{\mu}\in ]0, \mu[$, we can construct $V_\varkappa$
such that $V_\varkappa$ can be extended into a holomorphic function of $r=\vert x\vert$ in the sector  $\Sigma(2\varkappa)=\{\Re r\geq 1;  \vert \Im r\vert< 2\varkappa  \Re r\}$, and, for any 
multi-index  $\alpha$, it  satisfies 
\begin{equation}
\label{eq:2.3710}
\Vert \partial_x^\alpha(V_\varkappa(x)-V(x))\Vert_{N\times N}={\mathcal O} (\langle x\rangle^{-\tilde{\mu}-\vert \alpha\vert} \varkappa^\infty).
\end{equation}

As in \cite{MRS}, we fix $\varkappa=h^{s}$ with $s\in]0,1[$. We denote by $K_{\varkappa}(z;h)$ the right hand side of \eqref{For2}
when we replace $V$ by $V_\varkappa$ in $P_1(h)$.  The operator
$P_1(h)=-h^2\Delta+V_\varkappa$ can  be distorded analytically into $\tilde P_1(h)=U_\nu P_1(h)(U_\nu)^{-1}$  (see \cite{MRS}).   Now  the proof of Lemma \ref{ext an}  shows
that  \eqref{sk1} and  \eqref{equuu}    hold  for $K_{\varkappa}(z;h)$. On the other hand,  using the resolvent identity and   we   show that
$$K_{\varkappa}(z;h)-K(z;h)={\mathcal  O}(\varkappa^\infty)={\mathcal  O}(h^\infty), \;\;\; \text{ uniformly on }  \Gamma_0.$$
Consequently, Lemma \ref{ext an}  remains true under the assumptions \eqref{DA1} and \eqref{deff}.

%\section{Further results and generalizations}\label{Further results}

%%%%%%%%%%%%%%%%%%%%%%%%%%%%%%%%%%%%%%%%%%%%%%%%%%%%%%%%%%%%%%%%%%%%%%%%%%%%%%%%%%%%%%%%%%%%%%%%%%%%%
%%%%%%%%%%%%%%%%%%%%%%%%%%%%%%%%%%% Appendices %%%%%%%%%%%%%%%%%%%%%%%%%%%%%%%%%%%%%%%%%%%%%%%%%%%%%%
%%%%%%%%%%%%%%%%%%%%%%%%%%%%%%%%%%%%%%%%%%%%%%%%%%%%%%%%%%%%%%%%%%%%%%%%%%%%%%%%%%%%%%%%%%%%%%%%%%%%%

\appendix

\section{Microhyperbolic functions}\label{fonctions microhyperboliques}

{In this section, we prove some technical lemmas on the notion of microhyperbolicity used in our proofs.}

%microhyperbolic functions that we use in our proofs. Let us start with the following remark concerning the local aspect of the notion of microhyperbolicity.

\begin{lemma}\label{A1}
Let $H\in C^{\infty}(\mathbb R^{2n};\mathcal{H}_N)$. The following statements are equivalents 
\begin{itemize}
\item[(1)] $H$ is microhyperbolic at $\rho_0\in \mathbb R^{2n}$ in the direction $T\in \mathbb R^{2n}$.
\item[(2)] $\langle T, \nabla_{\rho} H(\rho_0)\rangle_{|\ker H(\rho_0)}$ is strictly positive in the sense of hermitian matrices, i.e. there exists $C>0$ such that
\begin{equation}\label{eto}
\big( \langle T, \nabla_{\rho} H(\rho_0) \rangle w, w \big) \geq C |w|^2, \quad \forall w\in \ker H(\rho_0).
\end{equation}
\end{itemize}
\end{lemma}

\begin{proof}
Obviously (1) implies (2). 

Assume that (2) is satisfied and let us prove (1). Let $w=w_1+w_2\in \mathbb C^N $, with $w_1\in \ker H(\rho_0) $ and $w_2\in \ker H(\rho_0)^{\bot}$. We have :
\begin{eqnarray*}
\big( \langle T, \nabla_{\rho} H(\rho_0) \rangle w, w \big) &=& \big( \langle T, \nabla_{\rho} H(\rho_0) \rangle w_1, w_1\big) + \big( \langle T, \nabla_{\rho} H(\rho_0) \rangle w_2, w_2 \big) \\ 
&+& \sum_{i\neq j=1}^2 \big( \langle T, \nabla_{\rho} H(\rho_0) \rangle w_i, w_j \big) =: I_1 + I_2 +I_3.
\end{eqnarray*}
By hypothesis, $I_1$ satisfies  
\begin{equation}\label{er1}
|I_1| \geq C |w_1|^2.
\end{equation}
On the other hand, we have 
\begin{equation}\label{er2}
|I_2| \leq C' |w_2|^2 \quad \text{and} \quad |I_3| \leq C'' |w_1| |w_2| < \varepsilon C'' |w_1|^2 + \frac{C''}{\varepsilon} |w_2|^2,
\end{equation}
for $\varepsilon>0$ small enough and $C',C''>0$. Putting together (\ref{er1}), (\ref{er2}), we obtain 
\begin{equation}\label{unq}
(\langle T, \nabla_{\rho} H(\rho_0)  \rangle w,w) \geq \frac{C}{2} |w|^2 - \mathcal{O}\left(\frac{1}{\varepsilon}\right) |w_2|^2.
\end{equation}
Now, the fact that $H(\rho_0) : \ker H(\rho_0)^{\bot} \rightarrow \ker H(\rho_0)^{\bot}$ is bijective, implies
$$
|H(\rho_0)w_2| \geq \tilde{C} |w_2|, \quad \forall w_2 \in \ker H(\rho_0)^{\bot}.
$$
Combining this with (\ref{unq}), we get 
$$
( \langle  T, \nabla_{\rho} H(\rho_0)  \rangle w, w  ) \geq \frac{C}{2}|w|^2 - \mathcal{O}\left(\frac{1}{\varepsilon}\right) |H(\rho_0)w_2|^2,
$$
which together with the fact that $H(\rho_0)w_2=H(\rho_0)w$ implies (\ref{eto}).
\end{proof}

\begin{lemma}\label{A2}
Let $F \in C^{\infty}(\mathbb R^{2n}; \mathcal{H}_{N-r})$ and $m(0)\in \mathcal{H}_r$ invertible, $r\geq 1$. Assume that for $\rho_0 \in \mathbb R^{2n}$, there exists $T\in \mathbb R^{2n}$ and $C_0>0$ such that 
\begin{equation}\label{rt}
\big( \langle T, \nabla_{\rho} F(\rho_0) \rangle w ,w  \big) \geq C_0 |w|^2, \quad \forall w \in \mathbb C^{N-r}.
\end{equation}
Then
\begin{itemize}
\item[(1)] $H(\rho)=\mleft(
\begin{array}{cc}
  F(\rho) & 0 \\
  0 & m(0)
\end{array}
\mright)$ is microhyperbolic at $\rho_0$ in the direction $T$. 
\item[(2)] If (\ref{rt}) holds at $\rho_0=0$ and $M\in C^{\infty}(\mathbb R^{2n},\mathcal{H}_N)$ with $M(\rho)= \begin{pmatrix}
\mathcal{O}(|\rho|^2) & \mathcal{O}(|\rho|) \\
\mathcal{O}(|\rho|) & \mathcal{O}(|\rho|)
\end{pmatrix}$, then $H+M$ is microhyperbolic near $\rho_0=0$ in the direction $T$.
\end{itemize}
\end{lemma}

\begin{proof}
Since $\langle T, \nabla_{\rho} H(\rho_0) \rangle = \begin{pmatrix}
\langle T, \nabla_{\rho} F(\rho_0) \rangle & 0 \\
0 & 0  
\end{pmatrix}$ and $\ker H(\rho_0) \subset \mathbb C^{N-r} \times {\{0_{r}\}}$, (1) follows immediately from Lemma \ref{A1}.

We have 
$$
\langle T, \nabla_{\rho} H(0)  \rangle + \langle T, \nabla_{\rho} M(0)  \rangle = \begin{pmatrix}
\langle T, \nabla_{\rho} F(0) \rangle & \mathcal{O}(1) \\
\mathcal{O}(1) & \mathcal{O}(1) 
\end{pmatrix}.
$$ 
Therefore
$$
\left(\langle  T, \nabla_{\rho} H(0) \rangle w ,w \right) + \left( \langle T, \nabla_{\rho} M(0) \rangle w ,w \right) = (\langle T, \nabla_{\rho} F(0)  \rangle w ,w) \geq C_0 |w|^2, \quad \forall w \in \mathbb C^{N-r}.
$$
Since $\ker \left( H(0) + M(0) \right) = \ker \left( H(0) \right)\subset \mathbb C^{N-r}$, it follows from lemma \ref{A1} that $H+ M$ is microhyperbolic at $\rho_0=0$ in the direction $T$. Then, $H+ M$ is microhyperbolic near $0$ in the direction $T$.
\end{proof}

The main result of this appendix is the following.
\begin{theorem}\label{main result appendix}
Let $H\in C^{\infty}(\mathbb R^{2n};\mathcal{H}_N)$. 
Assume that $H$ is microhyperbolic near $\rho_0\in \mathbb R^{2n}$ in the direction $T$.  
%Assume in addition that $\{\rho \in \mathbb R^{2n}; {\rm det}(H(\rho)-\tau_0)=0\}$ is a comact set.
 There exists $\tilde{H}\in C^{\infty}(\mathbb R^{2n};\mathcal{H}_N)$ such that $\tilde{H}=H$ near $\rho_0$ and $\tilde{H}$ is uniformly microhyperbolic on $\mathbb R^{2n}$ in the direction $T$. Moreover, we can choose $\tilde{H}$ bounded together with all its derivatives, i.e. $\tilde{H}\in S^0(\mathbb R^{2n};\mathcal{H}_N)$.
\end{theorem}
\begin{proof}
Without any loss of generality, we may assume that $\rho_0=0$. We know that there exists $P$ such that 
\begin{equation}\label{ute}
P H(0) P^{-1} = \begin{pmatrix}
0 & 0 \\
0 & m_{22}
\end{pmatrix},
\end{equation}
where $m_{22}$ is a diagonal and invertible matrix. Replacing $H(\rho)$ by $P H(\rho) P^{-1}$, we may assume that 
$$
H(\rho) = \begin{pmatrix}
m_{11}(\rho) & m_{21}(\rho) \\
m_{12}(\rho) & m_{22}(\rho)
\end{pmatrix}
$$
with $m_{11}(0)=0$, $m_{12}(0)=0$, $m_{21}(0)=0$ and $m_{22}(0)=m_{22}$. Since $H$ is microhyperbolic at $0$ in the direction $T$, it follows from Lemma \ref{A1} that 
\begin{equation}\label{gop}
\big\langle   \begin{pmatrix}
\langle T, \nabla_{\rho} m_{11}(0) \rangle & \langle T, \nabla_{\rho} m_{21}(0) \rangle \\
\langle T, \nabla_{\rho} m_{12}(0) \rangle &  \langle T, \nabla_{\rho} m_{22} \rangle
\end{pmatrix}  \left( \begin{array}{c}
w_1 \\
0 \\
\end{array} \right) , \left( \begin{array}{c}
w_1 \\
0 \\
\end{array} \right) \big\rangle = \big( \langle  T, \nabla_{\rho} m_{11}(0) \rangle w_1, w_1 \big) \geq C |w_1|^2
\end{equation}
We recall that $\ker\left( H(0) \right) \subset \{(w_1,0);\; w_1\in \mathbb C^{N-r}\}$, with $r= \dim  {\rm Im} \left( m_{22}\right)$ (due to (\ref{ute})).

Set
$$
H_0({\rho}) := \begin{pmatrix} {\nabla_{\rho}} m_{11}(0) \rho & 0 \\
0 & m_{22}
\end{pmatrix}, \quad \rho \in \mathbb R^{2n}.
$$
It follows from Lemma \ref{A2} and (\ref{gop}) that $H_0$ is microhyperbolic at every point $\rho \in \mathbb R^{2n}$ in the direction $T$. Let $\chi \in C_0^{\infty}(\mathbb R^{2n};\mathbb R)$ be such that $\chi(\rho)=1$ for $|\rho|\leq 1$ and $\chi(\rho)=0$ for $|\rho|\geq 2$. For $\delta>0$, set $\chi_{\delta}(\rho)= \chi\left( \frac{\rho}{\delta}  \right)$. We define 
$$
H_{\delta}({\rho}) = \chi\big( \frac{\rho}{\delta}  \big) \big( H(\rho) - H_0(\rho) \big) + H_0(\rho).
$$
We claim that for $\delta$ small enough, $H_{\delta}$ is microhyperbolic at every point $\rho \in \mathbb R^{2n}$ in the direction $T$. In fact, for $|\rho|\leq \delta$, $H_{\delta}(\rho) = H(\rho)$ is microhyperbolic at $\rho_0=0$ and then at every $\rho\in \mathbb R^{2n}$ with $|\rho| \leq \delta$. For $|\rho|\geq 2\delta$, $H_{\delta}(\rho) = H_0(\rho)$ which is microhyperbolic at every point $\rho\in \mathbb R^{2n}$ in the direction $T$. For $\delta< |\rho| < 2\delta$, we have 
$$
H_{\delta}(\rho)= H_0(\rho) + \begin{pmatrix}
\mathcal{O}(|\rho|^2) & \mathcal{O}(|\rho|) \\
\mathcal{O}(|\rho|) & \mathcal{O}(|\rho|)
\end{pmatrix}.
$$
Thus, Lemma \ref{A2} implies that $H_{\delta}$ is microhyperbolic in the direction $T$ for $\delta$ small enough. Consequently $H_{\delta}$ is microhyperbolic at every point $\rho \in \mathbb R^{2n}$ in the direction $T$.   To see that we can  choose $\tilde H\in  S^0(\mathbb R^{2n};\mathcal{H}_N)$, let $f\in C^\infty(\mathbb R)$ such that $f(t)=t$ for $\vert t\vert <1$, $\vert f(t)\vert \geq 1$  on $\vert t\vert \geq 1$ and $f(t)$ is constant at $\pm\infty$. Put $\tilde H(x)=f(H_\delta(x))$. By the functional calculus of self-adjoint operator, it is  easy to check that  $\tilde H$ satisfies the desired properties.

\end{proof}

\begin{proposition}\label{Hyperb}
Let $H\in C^{\infty}(\mathbb R^{2n};\mathcal{H}_N)$, $\chi\in C^\infty_0({\mathbb R}^{2n})$ and $\tau_0\in\mathbb R$. Assume that $\tau_0-H(\rho)$ is microhyperbolic 
at  every $\rho\in {\rm supp}\chi$. Let $G(\rho,z)$ be an $N\times N$ {matrix-valued function} (not necessarily Hermitian) {smooth} with respect to $\rho$ and holomorphic with respect to $z$ in a neighborhood of $\tau_0$. Set, for $\pm\Im z>0$ respectively,
$$F_\pm(z)=\int_{{\mathbb R}^{2n}} (z-H(\rho))^{-1} G(\rho,z) (z-H(\rho))^{-1} \chi(\rho) d\rho.$$
Then, for real $\tau$ near $\tau_0$, the limit   $F_\pm(\tau\pm i0):=\lim_{\varepsilon\searrow 0} F_\pm(\tau\pm i\varepsilon)$ exists and $\tau \rightarrow F_\pm(\tau\pm i0)$ 
is smooth near $\tau_0$.
\end{proposition}

\begin{proof}
We consider $F_+$. The proof for $F_-$ is similar.  Decomposing $\chi$ into a finite sum of functions $\chi_i$ with small support, we may assume  using  Theorem \ref {main result appendix} that $\tau -H(\rho)$ is microhyperbolic  in the direction $T$ at every point $\rho\in \mathbb R^{2n}$ and $\tau$ near $\tau_0$. We may also assume that $G, H\in S^0(\mathbb R^{2n};\mathcal{H}_N)$.
%Let $\tilde H$, $\tilde G$ and $\tilde \chi$
Let $\tilde{H},\tilde{G}$ and $\tilde \chi$ be three almost analytic extensions of $H$, $G$  and $\chi$ respectively, which are bounded together with all their derivatives. Put 
$$
\tilde{H}(\rho,t):= \tilde{H}(\rho+it T), \quad \tilde{G}(\rho,t,z):= \tilde{G}(\rho+it T,z), \quad \tilde{\chi}(\rho,t) := \tilde{\chi}(\rho+it T), \quad {t\in \mathbb R}.
$$
We assert that  for  small enough   $\Im z\geq 0, t\geq 0$  with $\Im z+t>0$, there exist $C,c>0$ such that 
\begin{equation}\label{Garding1}
\Im ((z-\tilde H(\rho,t))\omega, \omega )+C t
\vert (z-\tilde H(\rho,t))\omega\vert^2 \geq
c(t +\Im z)\vert \omega\vert^2,  \quad \forall \omega\in \mathbb C^N.
\end{equation}
In fact
$$((z-\tilde H(\rho,t))\omega,\omega)=((z- H(\rho))\omega,\omega)-it(\langle T, 
{\nabla_{\rho}} H(\rho) \rangle\omega,\omega)+{\mathcal O}(t^2) \vert \omega\vert^2$$
and hence  the global microhyperbolic condition (see \eqref{MHC}) yields, for some $c,C_1,C_2>0$,
$$ \Im ((z-\tilde H(\rho,t))\omega,\omega)\geq (\Im z+ct)\vert \omega \vert^2
-{\mathcal O}(t)\vert(\Re z-H(\rho))\omega\vert^2+{\mathcal O}(t^2) \vert \omega\vert^2$$
$$\geq c(\Im z+t-C_1(\Im z)^2-C_2t^2)\vert \omega \vert^2
-{\mathcal O}(t)\vert(z-\tilde H(\rho,t))\omega\vert^2,$$
uniformly on $\{z\in \mathbb C; \Re z\in ]\tau_0-\eta,\tau_0+\eta[,\,\Im z>0\}$ for small enough $\eta$, 
and   \eqref{Garding1}  follows from this inequality.

Applying Cauchy-Schwarz inequality to the first term of \eqref{Garding1}, we easily obtain 
\begin{equation}\label{A120}
\Vert z-\tilde H(\rho,t)\Vert_{N\times N}+Ct\Vert z-\tilde H(\rho,t)\Vert_{N\times N}^2\geq c (\Im z+t)\vert \omega\vert^2, \, \forall \omega\in \mathbb C^N.
\end{equation}
This shows that  $ (z-\tilde H(\rho,t))^{-1}$ exists  and
\begin{equation}\label{A34}
\Vert (z-\tilde H(\rho,t))^{-1}\Vert_{N\times N}={\mathcal O}(\frac{1}{t}).
\end{equation}
for {$t>0, \Im z \geq 0 $}.  

For simplicity, assume $T=(1,0,\cdots ,0)$. Put $\rho=(\rho_1,\rho')$ and  fix $t_0>0$. By the Stokes' formula, we have
$$F_+(z)=\int_{{\mathbb R}^{2n} }(z-\tilde H(\rho_1+it_0,\rho'))^{-1} \tilde G(\rho_1+it_0,\rho',z) (z-\tilde H(\rho_1+it_0,\rho'))^{-1}\tilde\chi(\rho_1+it_0,\rho') d\rho $$
$$-\iint_{{\mathbb R}^{2n}\times [0,t_0]} \frac{1}{2} (\partial_{\rho_1}+i\partial_t)\Big[(z-\tilde H(\rho,t))^{-1} \tilde G(\rho,z,t) (z-\tilde H(\rho,t))^{-1}\tilde\chi(\rho,t) \Big]  dt d\rho.$$
Clearly the first term of the right hand side of the above equality extends to a $C^\infty$ function up to $\Im z\ge 0$.
One sees that the same is true for the second term
by  using  \eqref{A34} and the fact that $(\partial_{\rho_1}+i\partial_t)\tilde H,   (\partial_{\rho_1}+i\partial_t) \tilde G,  (\partial_{\rho_1}+i\partial_t) \tilde \chi$ are all of
${\mathcal  O} (t^\infty)$. This ends the proof.
\end{proof}

\textbf{Acknowledgement.}  
This research was initiated when the first and second authors was visiting the Ritsumeikan University in   May  2016; the financial support and kind hospitality are gratefully acknowledged. The first author acknowledges support from JSPS KAKENHI Grant number JP16H03944. The third author was partially supported by the JSPS KAKENHI Grant number JP15K04971.

%Part of this work was carried out while the first and second authors visited the Ritsumeikan University, Kyoto, in May 2016. We are grateful to the Ritsumeikan university for the hospitality. This work was partially supported by the JSPS KAKENHI Grant number JP16H03944 (first author), by the Ritsumeikan university (second author) and by the JSPS KAKENHI Grant number JP15K04971 (third author).


\begin{thebibliography}{99}



%\bibitem{Ag}
%\textsc{J. Aguilar, J. M. Combes},
%\textit{A class of analytic perturbation for one-body Schr\"odinger Hamiltonians},
%Comm. Math. Phys, \textbf{22} (1971), 269-279.

%\bibitem{Assel} \textsc{R. Assel, M. Dimassi},
%\textit{Spectral Shift Function for the perturbations of Schr\"odinger operators at high energy},
%Serdica Math. J., \textbf{34} (2008), 253-266.


\bibitem{Birman}
    \textsc{M. S. Birman, M. G, Krein},
    \textit{On the theory of wave operators and scattering operators},
   Dokl. Akad. Nauk SSSR, \textbf{144} (1962), 475-478.


\bibitem{Birman1}
    \textsc{M. S. Birman, D. R. Yafaev},
    \textit{On the trace-class method in potential scattering theory},
   J. Soviet Math, \textbf{56}, no 2 (1993), 2285-2299.

\bibitem{Birman2}
    \textsc{M. S. Birman, D. R. Yafaev},
    \textit{The spectral shift function. The work of M. G. Krein and its further development},
   Algebra i Analiz \textbf{4}, no 5 (1992), 1-44, English trans in St Petersburgh Math J. \textbf{4}, no 5 (1993).

\bibitem{BR}
\textsc{V. Bruneau, D. Robert},
\textit{ Asymptotics of the scattering phase for the Dirac operator" High energy, semi-classical and non-relativistic limits},
Ark. Mat.  37, (199), 1--32.

\bibitem{Com}
    \textsc{J. M. Combes, P. Duclos, R. Seiler},
    \textit{The Born-Oppenheimer Approximation},
    ed. C. G. Velo and A. Wightman, Plenum Press, New York, (1981).

  \bibitem{CR}
    \textsc{M. Combescure, D. Robert},
    \textit{Coherent states and application in Mathematical Physics},
   Theoretical and Mathematical Physics (2012).
        

\bibitem{dimassi}
    \textsc{M. Dimassi, J. Sj\"ostrand},
    \textit{Spectral asymptotics in the semi-classical limit},
    London Mathematical Society, Lecture Note Series \textbf{268} (1999).


%\bibitem{Dimassi1}
%\textsc{M. Dimassi, M. Zerzeri},
%\textit{A time-independent approach for the study of spectral shift function},
%C. R. Acad. Sci. Paris, Ser. I \textbf{350} (2012), 375-378.


\bibitem{Dimassi2}
\textsc{M. Dimassi, S. Fujii\'e},
\textit{A time-independent approach for the study of the spectral shift function and an application to Stark Hamiltonians},
Comm. in Part. Diff. Equ., vol \textbf{40}, Issue 10 (2015), 1787-1814.

%\bibitem{Dimassi3}    
   %\textsc{M. Dimassi, J. Sj\"ostrand}, 
    %\textit{Trace asymptotics via almost analytic extension},
 %   Partial differential equations and Mathematical physics. Prog. Nonlin. Diff. Eq. Appl. \textbf{21} (1996), 126-142, Birkh\"auser, Boston.


%\bibitem{Clo}
%\textsc{T. Duyckaerts, C. Fermanian-Kammerer, T. Jecko},
%\textit{Degenerate codimension 1 crossings and resolvent estiamtes},
%Asymp. Anal. \textbf{65} (2009), no 3-4, 147-174.

%\bibitem{Faure}
%\textsc{F. Faure, J. Sj\"ostrand},
%\textit{Upper bound on the density of Ruelle resonances for Anosov flows},
%Comm. Math. Phys. \textbf{308} no. 2, (2011), 325-364.
 

\bibitem{GM}
    \textsc{C. G\'erard, A. Martinez},
    \textit{Principe d'absorption limite pour des op\'erateurs de Schr\"odinger \`a longue port\'ee},
    C.R. Acad. Sci. \textbf{306} (1987), 121-123.  

\bibitem{Hag}
\textsc{G. A. Hagedorn},
\textit{Molecular propagation through electron energy level crossing},
Mem. Amer. Math. Soc. \textbf{111} (1994), no 536.


%\bibitem{Helf}
%\textsc{B. Helffer, D. Robert},
%\textit{Calcul fonctionnel par la transformation de Mellin et op\'erateurs admissibles}.
%J. Funct. Anal. \textbf{53} (1983), no 3, 246-268.

\bibitem{Hormander}
\textsc{L. H\"ormander}, 
\textit{The analysis of linear partial differential operators},
vol I-IV, Springer Verlag Berlin Heidelberg New York (1985)

\bibitem{Hormander2}
\textsc{L. H\"ormander}, 
Fourier integral operators I, Ada Math. 127 (1971), 79-183.
%\bibitem{Hunz}
%\textsc{W. Hunziker},
%\textit{Distortion analyticity and molecular curves},
%Annales de L'I.H.P (Section Physique Th\'eorique) (1986).


\bibitem{Iso}
    \textsc{H. Isozaki, H. Kitada},
    \textit{Modified wave operators with time independent modifiers},
    J. Math. Phys \textbf{7} (1983), 137-143.   


\bibitem{iv}
    \textsc{V. Ivrii}, 
    \textit{Microlocal analysis and precise spectral asymptotics}, 
    Springer-Verlag, Berlin, (1998).
   

\bibitem{Jecko}
    \textsc{T. Jecko},
    \textit{Estimation de la r\'esolvante pour une mol\'ecule diatomique dans l'approximation de Born-Oppenheimer},
    Comm. Math. Phys. \textbf{195} (3) (1998), 585-612.        
 
%\bibitem{Jecko2}
%\textsc{T. Jecko},
%\textit{Non-trapping condition for semiclassical Schr\"odinger operators with matrix-valued potentials},
% Math. Phys. Elec. J. vol \textbf{11}, (2005)
 
%\bibitem{Jecko1}
 %   \textsc{T. Jecko},
  %  \textit{Semiclassical resolvent estimates for Schr\"odinger matrix operators with eigenvalues crossing}, Math. Nach. vol 257, Issue 1 (2003), 36-54.
    
\bibitem{Kato}
\textsc{T. Kato},
\textit{Perturbation theory for linear operators},
Springer-Verlag (1995).
 
\bibitem{Klein}
\textsc{M. Klein, A. Martinez, R. Seiler, X. P. Wang},
\textit{On the Born-Oppenheimer expansion for polyatomic molecules},
Comm. Math. Phys. \textbf{143} (1992), 607-639.

\bibitem{Khoc}
\textsc{A. Khochman},
\textit{Resonances and spectral shift function  for the semi-classical Dirac operator},
Rev. Math. Phys., 19 (2007), 1071-115.




Read More: http://www.worldscientific.com/doi/abs/10.1142/S0129055X0700319X?journalCode=rmp
    
\bibitem{Krein2}
\textsc{M. G. Krein},
\textit{On certain new studies in the perturbation theory for self-adjoint operators},
Topics in Differential and Integral Equations and Operators Theory, I Gohberg, Birkha\"user, Basel (1983).

\bibitem{Krein3}
\textsc{M. G. Krein},
\textit{On the trace formula in perturbation theory},
Mat. Sb, \textbf{33} (75) (1953), 597-626 (in russian).

\bibitem{Krein1}
\textsc{M. G. Krein, V. A. Javrjan},
\textit{On spectral shift functions arising in perturbation of a positive operator},
J. Operator. Th. \textbf{6} (1981), 155-191 (in russian).

\bibitem{Lif}
\textsc{I. M. Lifshits},
\textit{On a problem of perturbation theory},
Uspekhi Mat. Nauk. \textbf{7}, no 1, \textbf{143} (1952), 171-180. 
 
\bibitem{Lif1}
\textsc{I. M. Lifshits},
\textit{Some problems of the dynamic theory nonideal crystal lattices},
Nuovo Cimento Suppl. \textbf{3} (1956), 716-734.
 
\bibitem{MRS}   
 \textsc{A. Martinez, T. Ramond, J. Sj\"ostrand},   
\textit{Resonance for non-analytic potentials},    
    Analysis and PDE, vol 2 (2009) No 1, 29-60.


%\bibitem{Ned1}
%\textsc{L. Nedelec},
%\textit{Resonances for matrix Schr\"odinger operators},
%Duke Math. J., vol 106, \textbf{2}, (2001), 209-236.

%\bibitem{Ned2}
%\textsc{L. Nedelec},
%\textit{R\'esonances pour l'op\'erateur de Schr\"odinger matriciel},
%Ann. IHP, vol 65, \textbf{2}, (1996), 129-162.
    
\bibitem{Reed}
\textsc{M. Reed, B. Simon},
\textit{Methods of Modern Mathematical Physics},
Tome II, Fourier Analysis, Self-adjointness, Academic Press (1979).
   
\bibitem{Rob1}
    \textsc{D. Robert},
    \textit{Autour de l'approximation semi-classique},
    Progress in Mathematics \textbf{68} (1987), Birkh\"auser.   
 
\bibitem{Robert0}
   \textsc{D. Robert},
   \textit{Asymptotique de la phase de diffusion \`a haute \'energie pour des perturbations du second ordre du laplacien},
   Ann. scient. Ec. Norm. Sup. 4\`eme s\'erie, t.25 (1992), 107-134.

\bibitem{Robert1}
    \textsc{D. Robert},
    \textit{Semiclassical asymptotics for the Spectral Shift Function},
    Amer. Math. Soc. Trans. \textbf{189} (1999).    

\bibitem{Robert2}
    \textsc{D. Robert},
    \textit{Relative time-delay for perturbations of elliptic operators and semiclassical asymptotics},
    J. Funct. Anal \textbf{126} (1994), 36-82.    
    
\bibitem{Rob}
\textsc{D. Robert, H. Tamura},
\textit{Semiclassical bounds for resolvents of Schr\"odinger operators and asymptotics for scattering phases},
Comm. Part. Diff. Equ. \textbf{9}, no 10, (1984), 1017-1058.   

\bibitem{Robert}
    \textsc{D. Robert, H. Tamura},
    \textit{Semiclassical asymptotics for local spectral densities and time delay problems in scattering processus},
    J. Funct. Anal. \textbf{80} (1) (1988), 124-147.        

%\bibitem{Sj2}
%\textsc{J. Sj\"ostrand},
%\textit{A trace formula and review of some estimates for resonances, in Micorlocal Analysis and Spectral Theory (Lucca, 1996), NATO Adv. Sci. Inst. Ser. C Math. Phys. Sci \textbf{490}, Kluwer, Dordrecht, 1997}.

%\bibitem{Sj1}
%\textsc{J. Sj\"ostrand, M. Zworski},
%\textit{Complex scaling and the distribution of scattering poles},
%J. Amer. Math. Soc. \textbf{4} (1991), 729-769.

\bibitem{Sj}
    \textsc{J. Sj\"ostrand, M. Zworski},
    \textit{Fractal upper bounds on the density of semiclassical resonances},
    Duke Math, J. \textbf{137} (2007), no 3, 381-459.
    
    \bibitem{Wa1}
    \textsc{X. P. Wang}, Time-decay of scattering solutions and resolvent estimates for semiclassical Schr\"odinger operators, J. Diff. Equations, 71(1988), 348-395.
    
     \bibitem{Wa2}
    \textsc{X. P. Wang},   
   Semiclassical resolvent estimates for N-body  Schr\"odinger  operators, J. of Funct. Analysis, 97(1991), 466-483.

  \bibitem{Wa3}
    \textsc{X. P. Wang}, Time-decay of scattering solutions an classical trajectories. Ann. Inst. H. Poincar\'e, A, 47(1987), 25-37.
    
\bibitem{Yafaev}
\textsc{D. Yafaev},
\textit{Mathematical Scattering Theory}, 
Translation of Mathematical monographs, AMS \textbf{105}, Providence RI, (1992).


\end{thebibliography}
\end{document}